\documentclass[conference]{IEEEtran}

\ifCLASSINFOpdf
\else
\fi

\usepackage{epsfig,endnotes,xspace,url,diagbox}
\usepackage[bookmarks=false]{hyperref}
\usepackage{amsmath,amsthm,amsfonts}

\usepackage{enumerate}
\usepackage{varioref}
\usepackage{graphicx}

\usepackage{caption}
\usepackage{subcaption}
\usepackage{float}
\usepackage{epstopdf}
\usepackage{xcolor}
\usepackage{multirow}
\usepackage{comment}
\usepackage{color}
\usepackage[utf8]{inputenc}
\usepackage{mathtools}
\usepackage{wrapfig}
\usepackage{tabularx}
\usepackage{xspace}
\usepackage[noend,ruled,linesnumbered]{algorithm2e}

\usepackage{cases}
\usepackage{stfloats}
\usepackage{fancyhdr}
\usepackage{booktabs}

\usepackage{enumitem}
\setlist[itemize]{leftmargin=*}

\usepackage[numbers,sort&compress]{natbib}

\hypersetup{
  colorlinks,
  linkcolor={blue!70!green},
  citecolor={green!70!blue},
  urlcolor={orange!70!red}
}

\newcommand{\etal}{\textit{et al.}\xspace}
\newcommand{\ie}{\textit{i.e.}\xspace}
\newcommand{\eg}{\textit{e.g.}\xspace}

\newcommand{\etc}{\textit{etc.}\xspace}

\newcommand{\mypara}[1]{\noindent\textbf{#1.} \xspace}

\newtheorem{theorem}{Theorem}
\newtheorem{lemma}[theorem]{Lemma}

\definecolor{revision}{RGB}{0,0,255}

\newcommand{\revstart}{\begin{color}{revision}}
\newcommand{\revend}{~\!\!\end{color}}

\newcommand{\method}{\textsc{Victor}\xspace}
\newcommand{\mlda}{{ML-DA}\xspace}
\newcommand{\mt}{{MT}\xspace}

\usepackage{etoolbox}
\makeatletter
\patchcmd{\hyper@makecurrent}{%
    \ifx\Hy@param\Hy@chapterstring
        \let\Hy@param\Hy@chapapp
    \fi
}{%
    \iftoggle{inappendix}{%
        \@checkappendixparam{chapter}%
        \@checkappendixparam{section}%
        \@checkappendixparam{subsection}%
        \@checkappendixparam{subsubsection}%
        \@checkappendixparam{paragraph}%
        \@checkappendixparam{subparagraph}%
    }{}%
}{}{\errmessage{failed to patch}}

\newcommand*{\@checkappendixparam}[1]{%
    \def\@checkappendixparamtmp{#1}%
    \ifx\Hy@param\@checkappendixparamtmp
        \let\Hy@param\Hy@appendixstring
    \fi
}
\makeatletter

\newtoggle{inappendix}
\togglefalse{inappendix}

\apptocmd{\appendix}{\toggletrue{inappendix}}{}{\errmessage{failed to patch}}

\begin{document}

\title{\Large \bf \method: Dataset Copyright Auditing in Video Recognition Systems}

\author{
}

\author{\IEEEauthorblockN{Quan Yuan\IEEEauthorrefmark{1},
Zhikun Zhang\IEEEauthorrefmark{1}\IEEEauthorrefmark{2}\thanks{\IEEEauthorrefmark{2}Zhikun Zhang is the corresponding author.},
Linkang Du\IEEEauthorrefmark{3}, 
Min Chen\IEEEauthorrefmark{4}, \\
Mingyang Sun\IEEEauthorrefmark{5}, 
Yunjun Gao\IEEEauthorrefmark{1}, 
Shibo He\IEEEauthorrefmark{1},
Jiming Chen\IEEEauthorrefmark{1}\IEEEauthorrefmark{6}}

\IEEEauthorblockA{\IEEEauthorrefmark{1}Zhejiang University, %
\IEEEauthorrefmark{3}Xi’an Jiaotong University, 
\IEEEauthorrefmark{4}Vrije Universiteit Amsterdam, \\
\IEEEauthorrefmark{5}Peking University,
\IEEEauthorrefmark{6}Hangzhou Dianzi University\\ 
Email: 
\IEEEauthorrefmark{1}\{yq21, zhikun, gaoyj, s18he, cjm\}@zju.edu.cn, %
\IEEEauthorrefmark{3}linkangd@gmail.com, 
\IEEEauthorrefmark{4}m.chen2@vu.nl,
\IEEEauthorrefmark{5}smy@pku.edu.cn
}}

\IEEEoverridecommandlockouts
\makeatletter\def\@IEEEpubidpullup{6.5\baselineskip}\makeatother
\IEEEpubid{\parbox{\columnwidth}{
		Network and Distributed System Security (NDSS) Symposium 2026\\
		23 - 27 February 2026 , San Diego, CA, USA\\
		ISBN 979-8-9919276-8-0\\  
		https://dx.doi.org/10.14722/ndss.2026.240746\\
		www.ndss-symposium.org
}
\hspace{\columnsep}\makebox[\columnwidth]{}}

\maketitle 

\begin{abstract}

Video recognition systems are increasingly being deployed in daily life, such as content recommendation and security monitoring.
To enhance video recognition development, many institutions have released high-quality public datasets with open-source licenses for training advanced models.
At the same time, these datasets are also susceptible to misuse and infringement.
Dataset copyright auditing is an effective solution to identify such unauthorized use.
However, existing dataset copyright solutions primarily focus on the image domain; the complex nature of video data leaves dataset copyright auditing in the video domain unexplored.
Specifically, video data introduces an additional temporal dimension, 
which poses  significant challenges to the effectiveness and stealthiness of existing methods.

In this paper,
we propose~\method,
the first dataset copyright auditing approach for video recognition systems.
We develop a general and stealthy sample modification strategy that enhances the output discrepancy of the target model.
By modifying only a small proportion of samples (\eg, 1\%),
\method amplifies the impact of published modified samples on the prediction behavior of the target models.
Then, the difference in the model's behavior for published modified and unpublished original samples can serve as a key basis for dataset auditing.
Extensive experiments on multiple models and datasets highlight the superiority of~\method.
Finally, we show that \method is robust in the presence of several perturbation mechanisms to the training
videos or the target models.

\end{abstract}

\section{Introduction}
\label{sec:introduction}

Video recognition systems~\cite{wu2022survey} have become increasingly vital in real-world applications, including content recommendation on streaming platforms~\cite{deldjoo2016content}, activity recognition in surveillance systems~\cite{sun2022human}, autonomous driving~\cite{biparva2022video} and healthcare monitoring~\cite{elharrouss2021review},~\etc
To facilitate the research of video recognition,
many high-quality datasets are published by the related research institutions such as DeepMind~\cite{kay2017kinetics}, The Allen Institute for AI~\cite{sigurdsson2016hollywood}, and Stanford University~\cite{karpathy2014large}.
However, it should be noted that the video datasets are published with strict open-source licenses to protect the intellectual property of the data owner.
Unauthorized commercial deployment would not only easily violate these terms, but also raise serious legal and ethical issues.
According to an investigation by Proof News~\cite{2024youtube},
many companies used YouTube video datasets for AI training without permission, which conflicts with YouTube’s terms of service.
Therefore, it is crucial to audit the dataset copyright in video recognition systems.

Currently, existing dataset auditing research primarily focuses on image~\cite{sablayrolles2020radioactive,li2023black,guo2023domain} and audio~\cite{guo2025audio,miao2021audio} domains,
while few works have yet explored the dataset auditing for video recognition systems. 
Existing dataset auditing research can be divided into passive auditing and proactive auditing~\cite{du2025sok,huang2024general}. 
Passive auditing, primarily built on  membership inference~\cite{shokri2017membership}, 
infers whether the data of a user is utilized to train the target model~\cite{song2019auditing,chen2023face}. 
However, this method requires a large number of queries to the model, which is costly. 
On the other hand, it usually leads to a high false alarm rate, which is impractical. 
In contrast, proactive auditing injects certain marks such as radioactive data and backdoor into some samples before the dataset is released for subsequent verification~\cite{sablayrolles2020radioactive,li2022untargeted}. 
However, most prior studies either require altering data labels, which can degrade the model's performance on normal tasks, or rely on the knowledge of target model, which is often impractical in real-world scenarios.
In addition, 
compared to image data, the additional spatiotemporal dimension of video data and the diversity of various video models pose significant challenges, making most existing methods difficult to implement or ineffective in practice.

\mypara{Our Proposal}
In this work, we present \method~(\underline{Vi}deo 
re\underline{c}ognition
Audi\underline{tor}), a practical approach for auditing the dataset copyright by modifying a small portion of the samples in the published dataset. 
The core idea of~\method is to amplify the impact of published modified samples on the prediction behavior of the target model.
By evaluating the difference in the outputs of the suspect model for the modified samples and the original samples, \method determines whether the dataset has been misused.
If the difference in the output prediction of the original sample and the modified sample is small, it indicates a high likelihood that the suspect model was trained using the published dataset.
The design of~\method mainly faces the following challenges:

\textit{Challenge 1: How to avoid introducing side effects on model training?}
One prevalent strategy for influencing model behavior is to embed backdoors into the trained model. 
However, such operations can easily cause backdoored samples to produce incorrect predictions, %
introducing unexpected or exploitable vulnerabilities.
Accordingly, \method chooses to retain the original true label of each modified sample throughout the modification process to prevent potential side effects.

\textit{Challenge 2: How to amplify the impact of modified samples with low modification costs? 
}
On the one hand, the dataset owners usually have limited information about the target model. %
On the other hand, video data is often subject to operations such as interception and cropping before being processed by the target model.
These characteristics significantly hinder the applicability of traditional image dataset auditing techniques to video-based systems.
To ensure that the amplification effect remains robust across diverse
target models 
and video preprocessing methods, \method introduces procedural noise~\cite{co2019procedural} into all frames of selected samples. 
This subtle perturbation effectively enhances the behavioral difference between the released and original samples on the target model, without compromising the visual semantics.

\textit{Challenge 3: How to achieve high-precision auditing with low false positive rates?}
On the one hand, it is challenging to design a reliable mechanism for detecting dataset misuse.
On the other hand, if both samples produce low confidence scores for the correct label, their outputs may appear similar, thus obscuring the basis for auditing and increasing the risk of misjudgment.
To enable reliable detection of dataset usage, \method employs a subset of published original samples and their corresponding unpublished modifications to estimate a decision threshold. 
The auditing is carried out via hypothesis testing on the sequence of output differences between unpublished original and published modified samples. 
Moreover,
\method introduces post-processing steps to address scenarios in which both types of sample produce low output probabilities, thereby reducing misleading in the auditing decision.

\mypara{Evaluation}
We conduct experiments on multiple representative
video recognition models and three classic open-sourced video datasets to
illustrate the effectiveness of~\method.
The experimental results indicate that \method can effectively audit the dataset usage across various settings.
For instance, by modifying only 1\% of the samples, \method can achieve up to 100\% auditing accuracy across 
multiple datasets and suspect models.
We further analyze the effectiveness of various components %
and explore the impact of different parameter settings.

\mypara{Robustness}
In practice, the target model might be equipped
with various obfuscation techniques to hinder the dataset auditing.
Therefore,
we conduct the experiments to validate the robustness of our proposed~\method. 
In our work, three representative interference strategies (\ie, 
input preprocessing,
training intervention, post-adjustment) in a general machine learning model pipeline are considered.
We observe that the performance of~\method only slightly drops, which shows the robustness of~\method.

\mypara{Contributions}
In summary, the main contributions of the paper are three-fold:

\begin{itemize}
[itemsep=2pt,topsep=2pt,parsep=0pt]

\item To our knowledge, ~\method is the first dataset copyright auditing approach for video recognition systems.

\item 
We propose a label-invariant perturbation mechanism and a behavior difference-based verification strategy to enable effective auditing.
In particular,
\method injects procedural noise into a small fraction (\eg, 1\%) of the dataset without altering the label. 
This modification amplifies the influence of altered samples on the target model while effectively keeping visual content and model utility. 
Combined with careful designs such as threshold estimation and hypothesis testing, \method establishes a reliable foundation for video dataset auditing.

\item We conduct comprehensive experiments on multiple models and datasets to illustrate the effectiveness and robustness of \method.
\method is open-sourced at 
\url{https://github.com/sec-priv/VICTOR}.%

\end{itemize}

\section{Related Work}

\label{sec:related_work}

The dataset copyright auditing methods can be divided into passive auditing and proactive auditing methods,
depending on whether the original dataset is modified.

\mypara{Passive Auditing}
The solutions in passive auditing are usually implemented based on the idea of membership inference attack~\cite{shokri2017membership,carlini2022membership,zhang2022inference,liu2022ml}.
The fundamental principle of membership inference is to identify distinct characteristics between data that has been used in training and data that has not.
The existing passive auditing methods can be classified into decision boundary-based and behavior characteristics-based methods~\cite{du2025sok}.
For the decision boundary-based methods~\cite{maini2021dataset,li2021membership,choquette2021label,tian2023knowledge,szyller2023robustness}, the core intuition is that samples near the decision boundary in the training dataset are critical for classification. Therefore, by extracting boundary information from the model, the dataset owner can infer whether a particular dataset was used in its training.
For the behavior characteristics-based methods~\cite{chen2023face,du2024orl,liu2021encodermi,song2021systematic,du2025artistauditor,carlini2022membership,li2025vid},
the dataset owner utilizes the model’s outputs and hidden representations as the key basis to finish the auditing.
In particular, the model's behavior include the
loss value~\cite{sablayrolles2019white}, the log-likelihood value~\cite{dziedzic2022dataset,li2022user},
and the disparity between the outputs of
the target and the shadow models
(\ie, models trained on datasets that are similar to the training dataset of the target model)~\cite{dong2023rai2,liu2022your,salem2019ml}.
However, these passive auditing approaches usually face significant challenges such as high false positive rates and the need for intensive
queries to target models.
In light of the low accuracy and poor robustness of passive auditing solutions, we focus on the proactive auditing strategy to safeguard the dataset copyright in this work.

\mypara{Proactive Auditing}
A typical solution in proactive auditing is radioactive data-based auditing~\cite{sablayrolles2020radioactive,wenger2024data,guo2023domain,guo2024zeromark}.
This type of method 
injects an optimized radioactive mark into the vanilla training images. In practice, the radioactive marks
need to be propagated to the image space.
If the marked data are used in the training, the classification model
is updated with both the features and the radioactive mark.
In the copyright validation, the auditor detects the distribution
deviation induced by the radioactive marks.
However, the distribution shift may be slight by a single marked sample. 
Thus, the dataset owner needs to inject
a large number of marked images into the original dataset to provide statistical evidence that the model is trained on marked images.
In addition, recent studies indicate that the performance of radioactive data-based methods is limited~\cite{huang2024general,chen2025MembershipTracker}.

Another classic method of proactive auditing is backdoor-based auditing.
In this setting, the dataset owner embeds backdoors (or called triggers) to the original dataset.
If a model is training using this modified dataset,
the model will perform normally for benign samples.
When specific backdoors are present,
the model's predictions will change dramatically.
Depending on whether the true labels of the samples are altered during the backdoor injection,
the backdoor-based method can be classified into dirty-label backdoor~\cite{gu2019badnets,li2021invisible,li2020open,li2023black,li2022black,al2024look} and clean-label backdoor methods~\cite{li2022untargeted,souri2022sleeper,tang2023did}.
The dirty-label backdoor methods  
have better effectiveness, and the clean-label backdoor methods achieve better stealthiness.
However,
both the two type of methods are prone to introduce harmful influence and potential security risks~\cite{bouaziz2025data}.
The attacker can exploit the backdoor to interfere with the model performance on other normal samples~\cite{guo2025audio}.

Recently, Huang~\etal~\cite{huang2024general} proposed a data auditing framework and applied it to image classifiers and foundation models. 
The proposed framework leverages the existing black-box membership inference method, together with a sequential hypothesis testing to implement dataset auditing.
However, this approach requires a substantial proportion of data marking in the published version. 
Chen~\etal~\cite{chen2025MembershipTracker} employ a data marking component to mark the target data and adopt a membership inference-based verification process for auditing. 
However, this approach requires substantial modifications to the original images, which significantly degrades image quality.
Moreover, the dynamic nature of video data makes the auditing performance of the above methods insignificant.

\section{Preliminaries}
\label{sec:preliminary}

\subsection{Video Recognition}
\label{subsec:video_classification_system}

Video recognition systems are designed to analyze and interpret the spatio-temporal content of videos, with the goal of assigning semantically meaningful labels to entire clips or localized temporal segments~\cite{sadhu2021visual}. 
This system has already been widely applied for various fields, such as action recognition, gesture analysis, healthcare monitoring, and content categorization~\cite{pareek2021survey}.

Before being fed into the models, raw video inputs typically undergo several pre-processing steps~\cite{parashar2023data}. These include frame extraction, spatial resizing, temporal sampling (\eg, uniform, random, or dense sampling), and modality transformation (\eg, computing optical flow or extracting RGB and depth channels). 
These pre-processing operations are essential for reducing computational overhead, maintaining temporal consistency, and enhancing the extraction of 
spatio-temporal features~\cite{liang2023survey}.

Over the past decade, researchers have proposed various model architectures for video recognition, which can be categorized into three types:
2D CNN + RNN, 3D-CNN, and Transformer-based models.
The 2D CNN + RNN methods utilize 2D CNN for frame feature extraction and RNN
for capturing temporal dependencies between them~\cite{lin2019tsm,luo2019grouped,wang2016temporal,liu2016spatio}.
To learn stronger spatial-temporal representations, 3D
CNN-based methods were proposed~\cite{feichtenhofer2020x3d,feichtenhofer2019slowfast,tran2018closer,tran2019video,carreira2017quo}. 
These methods utilize
3D kernels to jointly leverage the spatio-temporal context
within a video clip.
Inspired by the success in natural language processing, transformer-based models are proposed~\cite{liu2022video,bertasius2021space,arnab2021vivit}. 
This type of method adopts self-attention to gauge the relevance of various frames.

\subsection{Threat Model}
\label{subsec:problem_statement}

\mypara{Application Scenario}
This paper aims to address a critical gap in video dataset copyright protection: detecting unauthorized commercial use of open-source licensed datasets.
\autoref{fig:scenario}
illustrates a typical application scenario where the data owners 
publish the dataset to the public.
In order to protect the dataset copyright, the dataset owner modifies some samples (videos) in the original dataset. 
The modified samples and other original samples constitute the dataset $D_1$ and are released to the public.
A malicious dataset user (attacker) can download the open-sourced video datasets and then train their models for unauthorized purposes.
To confirm whether the trained models use the dataset $D_1$, 
the auditor
queries these models to obtain the auditing results.

\begin{figure}[htbp]
\centering
\vspace{-0.2cm}
\includegraphics[width=0.42\textwidth] {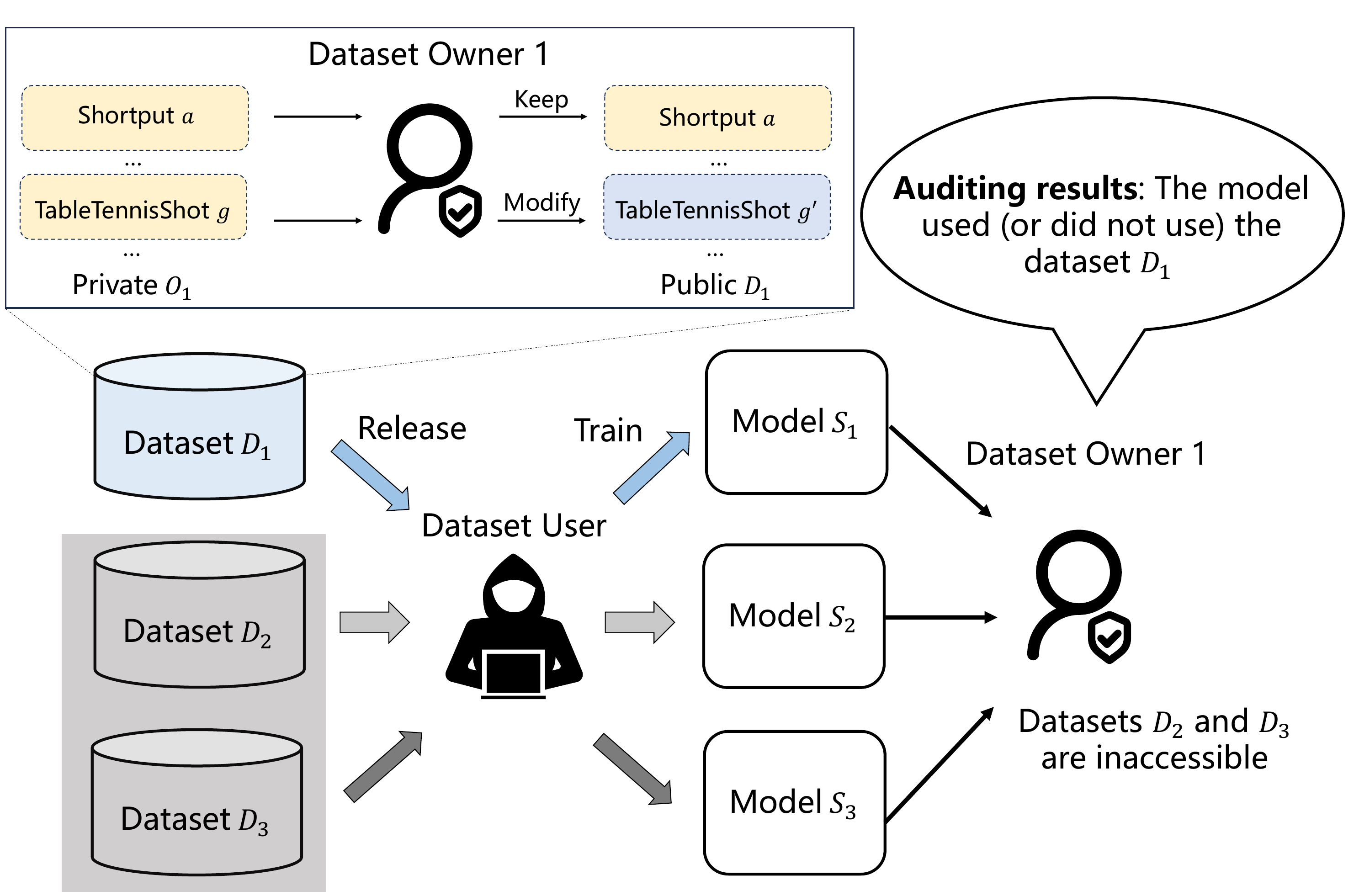}
\caption{An example of the application scenario. 
}
\vspace{-0.2cm}
\label{fig:scenario}
\end{figure}

\mypara{Attacker's Capabilities}
We assume that the attacker can download the published video dataset and train any video recognition model with the dataset.
The attacker may adopt various evasion strategies (such as input preprocessing, training intervention, and post-adjustment) to bypass the auditing mechanisms.

\mypara{Auditor’s Capabilities}
The auditor has complete knowledge of the protected dataset and can modify some samples of the dataset prior to release.
In addition to its own dataset, we consider that the auditor is unaware of the dataset information from other dataset owners.
To mimic the real-world application, we consider that
the auditor has black-box access to the suspect model.
Note that this is the most general and challenging scenario for the auditor.

\section{\method}
\label{sec:methodology}

\subsection{Can Image Auditing Methods Apply to Video?}
\label{subsec:challenge}

Before introducing our approach, one
might wonder why image auditing techniques cannot
be applied directly to video. 
The challenges faced when applying
the image auditing methods to video are as follows:

\mypara{Flexible Video Lengths}
The complexity of video data (including additional temporal dimension) is higher than image data.
In particular,
video data comprises sequences of frames with uncertain lengths, which are typically preprocessed through operations such as frame sampling and cropping before being input into neural networks.
Under these conditions, 
the difference in the amount of information between the raw data and the model input data is significant, making 
the 
traditional image auditing methods difficult to implement or ineffective.

\mypara{Complex Video Models}
Due to the increased architectural complexity of video recognition models and their reliance on aggregated information from multiple frames, image auditing techniques designed to perturb the output distribution are generally ineffective especially when the target model architecture and data pre-processing steps are unknown.
The greater robustness of video recognition models significantly increases the difficulty of influencing the model output.

\mypara{Harmful Backdoor Influence}
Current image auditing approaches based on backdoor injection typically introduce harmful triggers into the training data. 
These triggers may not only compromise the model's performance but also introduce exploitable vulnerabilities that attackers can leverage.

\subsection{Motivation}
\label{subsec:motivation}

Given that the dataset owner has the ability to modify the data prior to release,
our core idea is to amplify the behavioral differences of these modified samples across different target models (\ie, models that have used the dataset and those that have not).
During the copyright verification phase, these amplified differences can then be evaluated to determine whether the dataset has been used for training.

For the design goal of dataset auditing, we aim to achieve the following objectives:
1) Modify only a small fraction of the dataset at a low cost to reduce the discrepancies between the published and original datasets;
2) Ensure that the proposed method can reliably identify whether the released dataset was misused by leveraging the behavioral differences in the suspect model across various data pre-processing procedures and model architectures;
3) Avoid introducing harmful backdoor influence that could degrade the performance of the target model;
4) Maintain robustness against various adaptive attacks.

\subsection{Overview}
As depicted in \autoref{fig:framework}, the workflow of \method mainly consists of three phases: Sample modification, sample selection, and copyright verification. 
If the suspect model is trained on the target dataset, the auditor should output a positive auditing result;
otherwise, a negative auditing result.
We also summarize the frequently used mathematical notations in \autoref{table:math_notations}.

\begin{figure*}[htbp]
\centering
\includegraphics[width=0.92\textwidth]{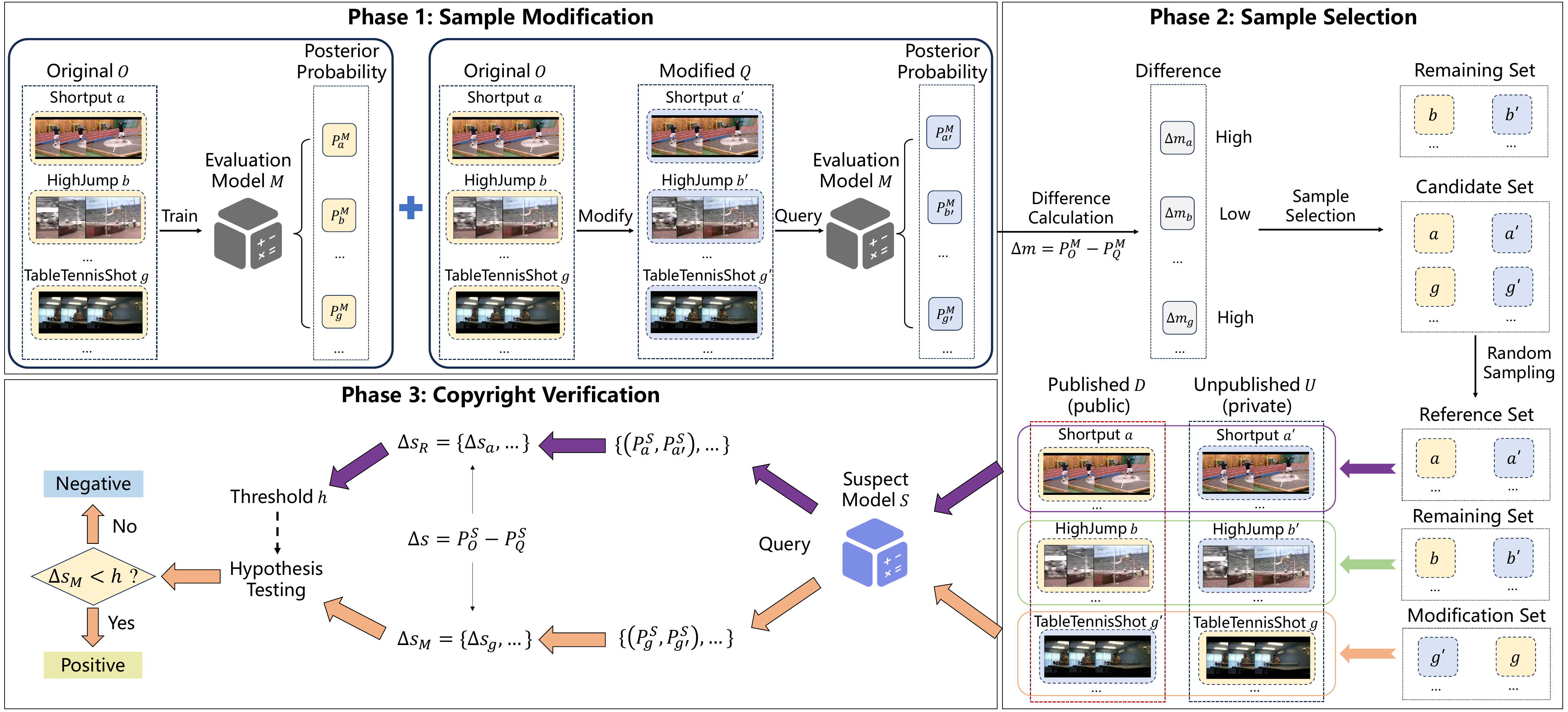}
\caption{\method overview. \method is composed of three phases: Sample modification, sample selection, and copyright verification.
In the sample modification phase, an evaluation model is trained based on the original dataset.
In addition, the modified version of original dataset is generated in this phase. 
In the sample selection phase, \method quantifies the difference between each modified sample and its original counterpart by analyzing the evaluation model's outputs.
This difference serves as a key criterion for selection.
Samples with larger differences form the candidate set, while those with smaller differences are placed in the remaining set. 
\method then randomly selects a subset of candidate samples as the modification set and another subset as the reference set.
\method proceeds to publish the modified samples from the modification set, along with the original samples from the reference and remaining sets.
In the copyright verification phase, \method assesses whether the dataset has been misused by analyzing the suspect model’s predictions on the reference and modification sets.
}
\vspace{-0.2cm}
\label{fig:framework}
\end{figure*}

\begin{table}[!t]
    \centering
    \caption{Summary of mathematical notations.}
    \label{table:math_notations}
    \vspace{-0.1cm}
    \footnotesize
    \setlength{\tabcolsep}{1.2em}
	\begin{tabular}{cl}
		\toprule
		\textbf{Notation} & \textbf{Description}  \\
		\midrule
            $O$ & Original dataset \\
		$Q$ & Modified dataset \\
            $D$ & Published dataset \\
		$U$ & Unpublished dataset \\
		$\varepsilon$ & Perturbation budget  \\
             $\delta_{\text{per}}$ & Parameter set for Perlin noise \\
		$M$ & Evaluation model \\
		$S$ & Suspect model  \\
		$P$ &  Posterior probability \\
		$\mathbb{R}$ & Reference set \\
		$\mathbb{M}$& Modification set \\
            $\mathbb{E}$& Remaining set \\
		$r_c$& Ratio of candidate set  \\
        $r_m$& Ratio of modification set  \\
        $r_r$& Ratio of reference set  \\
        $H$& Upper limit of the difference threshold  \\
        $B$& Boundary of low probability  \\
        $\beta$& Boundary coefficient  \\
        $\alpha$& Significant level  \\
		\bottomrule
	\end{tabular}
\end{table}

\mypara{Phase 1: Sample Modification}
In this phase,
an evaluation model is first trained based on the original video dataset.
In addition, we can generate the modified version for each sample (video) in the original dataset.
Then, the prediction outputs of original and modified versions can be calculated by the evaluation model.
These predictions can be served as the selection basis in the following phase.
The details of Phase 1 are in~\autoref{subsec:sample_modification}.

\mypara{Phase 2: Sample Selection}
Based on the output results from Phase 1, 
we hope to find the samples that are most likely to achieve amplification effect as the final published samples.
Here, we choose to calculate the difference in the predictions of the evaluation model on the original and modified datasets.
The samples with larger differences form the candidate set,
while those samples with smaller differences construct the remaining set.
Next, we select a subset of candidate samples as the modification set and another subset as the reference set.
Further, the other samples in the candidate set that were not selected are also included in the remaining set.
We then
release the public dataset, which includes the modified samples of the modification set, and the original samples of the reference and remaining sets.
The details of Phase 2 can be found in~\autoref{subsec:sample_select}.

\mypara{Phase 3: Copyright Verification}
During this phase, \method determines whether the dataset has been misused by leveraging the samples from the modification and reference sets.
The key insight is to assess whether there exists a statistically significant difference in the suspect model’s performance between the published and unpublished samples from the reference and modification sets.
The details of Phase 3 are referred to~\autoref{subsec:copyright_verify}.

\section{Design Details}
\label{sec:design_details}

\subsection{Sample Modification}
\label{subsec:sample_modification}

During the sample modification phase,
\method aims to obtain an evaluation model and a modified dataset based on the original dataset.
\autoref{algorithm:sample_modify} illustrates the basic process of this phase.

First, 
\method trains an evaluation model $M$ based on the original dataset $O$.
Then, the posterior probability prediction of each sample in $O$ can be computed according to the model $M$.
The evaluation model $M$ can serve as a basis for guiding sample selection. 
Intuitively, a larger prediction discrepancy between the original and modified samples on 
$M$ indicates a higher potential for amplification, making such samples more suitable for effective auditing.

In addition,
\method tries to generate a modified version for each sample in the original dataset $O$.
Our objective is to generate modified samples that remain visually similar to the original videos while keeping the label unchanged and inducing noticeable differences in the output of the target model, under constrained modification costs.
However, achieving this goal on video data presents several challenges.
First, although a video consists of many frames, only a subset is typically used for model prediction, and both the number and indices of these selected frames are often unknown.
Second, video recognition models aggregate temporal information across multiple frames during inference, making them inherently more robust than image models.
Furthermore, the architectures of target models are highly diverse, exacerbated by the difficulty of designing universally effective modifications.
These factors significantly limit the applicability of traditional image domain solutions in the video domain.

To address the above challenges, we draw inspiration from two-dimensional procedural noise~\cite{lagae2010survey}, especially Perlin noise~\cite{perlin2002improving}, which has been shown to effectively reduce prediction confidence in various image classification models~\cite{co2019procedural,chen2025MembershipTracker}. 
Due to the great versatility and low implementation cost of procedural noise,
we choose to inject it into the original video.
In particular, 
procedural noise refers to algorithmically generated patterns that simulate the randomness and irregularity found in natural phenomena. 
\autoref{fig:perlin_show1} illustrates the examples of procedural noise.

\begin{figure}[!t]
    \centering
    \includegraphics[width=0.35\textwidth]{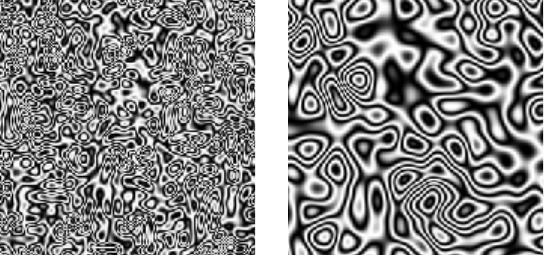}
    \caption{
    The example of Perlin noise.
    }
    \label{fig:perlin_show1}
\vspace{-0.4cm}
\end{figure}

\begin{algorithm}[!t]

        \caption{Sample Modification}
        \label{algorithm:sample_modify}
        \KwIn{Original dataset $O$, perturbation budget $\varepsilon$, the parameter set for three-dimensional Perlin noise $\delta_{\text{per}} = \{\lambda_x, \lambda_y, \lambda_t, \phi_{\text{sine}}, \Omega\}$
        }
        \KwOut{The evaluation model $M$, the modified dataset $Q$, the posterior probability vectors $P_O^M$ and $P_Q^M$}
        {
        // Model training \\
        Train an evaluation model $M$ based on the original dataset $O$  \\
        Calculate the posterior probability $P_O^M$ of the samples in the original dataset $O$ by evaluation model $M$
        \\

        // Noise injection \\
        \For{each sample o in O}
        {
        Generate the normalized three-dimensional Perlin noise with the parameter set $\delta_{\text{per}}$ \\
        Scale the noise based on the perturbation budget $\varepsilon$ \\
        Add the noise to the original sample and normalize it to $[0,255]$ to obtain the modified sample $o'$, then add $o'$ to $Q$ \\
        Calculate the posterior probability $P_{o'}^M$ of the modified sample and update it to $P_{Q}^M$
        }
       
        }
       
\end{algorithm}

In this work, 
we extend the classical two-dimensional Perlin noise construction to the spatio-temporal domain to generate dynamic, temporally-coherent perturbations suitable for video data. 
In addition to the spatial wavelengths $\lambda_x, \lambda_y$, and the number of octaves $\Omega$, we further introduce a temporal frequency control parameter $\lambda_t$ to regulate variation across the temporal axis.

Let $(x, y, t)$ be a query point in normalized coordinates, and let $(i,j,k)$ be the corresponding corner lattice points, \ie, 
$i = \lfloor x \rfloor, j = \lfloor y \rfloor, k = \lfloor t \rfloor$.
Each lattice corner is assigned a pseudo-random unit gradient vector:
\begin{equation*}
\mathbf{g}_{ijk} \sim \mathcal{N}(0, \mathbf{I}_3),
\end{equation*}
where $\mathcal{N}(\mathbf{0}, \mathbf{I}_3)$ denotes a 3-dimensional normal distribution with zero mean and identity covariance matrix, \ie, $\mathbf{g}_{ijk}$ is sampled independently from a standard multivariate Gaussian.
The relative offset from the lattice corner is:
\begin{equation*}
u = x - i, \quad v = y - j, \quad w = t - k. 
\end{equation*}

We compute the dot product between the gradient vector and the relative offset vector at each of the 8 corners of the unit cube, followed by trilinear interpolation using Perlin's fade function:
\begin{equation}
\label{eq:perlin_fade_func}
f(s) = 6s^5 - 15s^4 + 10s^3. 
\end{equation}

Let $P_n(x, y, t)$ denote the interpolated Perlin noise value. 
We then apply a fractal summation over $\Omega$ octaves:
\begin{equation}
\label{eq:S_per}
S_{\text{per}}(x, y, t) = \sum_{n=1}^{\Omega} \frac{1}{2^n} P_n\left(\frac{2^{n-1}x}{\lambda_x}, \frac{2^{n-1}y}{\lambda_y}, \frac{2^{n-1}t}{\lambda_t} \right).
\end{equation}

To enhance the %
perturbation diversity and control high-frequency content, we apply a sinusoidal transformation:
\begin{equation}
\label{eq:G_per}
G_{\text{per}}(x, y, t) = \sin\left(2\pi \cdot \phi_{\text{sine}} \cdot S_{\text{per}}(x, y, t)\right).
\end{equation}

Here, $\phi_{\text{sine}}$ is a tunable sine frequency parameter that controls the intensity and granularity of the generated texture. 
Therefore, the full parameter set for three-dimensional Perlin noise is
$
\delta_{\text{per}} = \{\lambda_x, \lambda_y, \lambda_t, \phi_{\text{sine}}, \Omega\}
$.

The final noise field is normalized to the range $[0, 1]$.
Then, we utilize the budget $\varepsilon$ as the $\ell_{\infty}$-norm constraint on the perturbation. 
Furthermore,
we inject the scaled noise into all frames of the original video and obtain the modified sample, with the constraint of $[0,255]$.
On this basis, we can obtain the posterior predictions of the modified samples.
In addition, we clarify that 3D Perlin noise can be replaced with other types of noise in practical deployment (\eg, Gabor noise), which demonstrates the generalizability and stealthiness of \method.

Through the above steps, in the first phase, we can obtain an evaluation model $M$ by training and generate a modified dataset $Q$, which can be used for subsequent selection.

\subsection{Sample Selection}
\label{subsec:sample_select}

\begin{algorithm}[!t]

        \caption{Sample Selection}
        \label{algorithm:sample_select}
        \KwIn{Original dataset $O$, the modified dataset $Q$, the posterior probability vectors $P_O^M$ and $P_Q^M$,
        the ratio of candidate set $r_c$, the ratio of reference set $r_{r}$,
        the ratio of modification set $r_{m}$
        }
        \KwOut{The published dataset $D$, the unpublished dataset $U$, the reference set $\mathbb{R}$, the modification set $\mathbb{M}$, the remaining set $\mathbb{E}$}
        {
        // Difference calculation \\
        $\Delta m \leftarrow P_O^M-P_Q^M$  \\
        Sort the values in $\Delta m$ 
        \\
        Select samples with values in the first $r_c$ proportion in $\Delta m$ as the candidate set $\mathbb{C}$, and the remaining 
$(1-r_c)\cdot|O|$ samples as the remaining set $\mathbb{E}$ \\

        // Random sampling \\
        Randomly select $r_r\cdot|O|$ samples from the candidate set $\mathbb{C}$ as the reference set $\mathbb{R}$, and select $r_m\cdot|O|$ samples as the modification set $\mathbb{M}$, where $\mathbb{R}$ and $\mathbb{M}$ do not overlap \\
        Add other $(r_c-r_m-r_r)\cdot|O|$ samples in the candidate set $\mathbb{C}$ that do not belong to $\mathbb{R}$ and $\mathbb{M}$ to the remaining set $\mathbb{E}$ \\
        // Dataset division \\
        The published dataset $D$ consists of modified samples of the modification set, as well as original samples of the remaining set and the reference set \\
        The unpublished dataset $U$ consists of original samples of the modification set, as well as modified samples of the remaining set and the reference set \\
       
        }
       
\end{algorithm}

Based on the modified samples obtained in the first phase, 
the published modified samples can be selected in this phase.
\autoref{algorithm:sample_select} provides the process of the sample selection. %

First, we compute the prediction difference $\Delta m$ between the original and modified datasets using the outputs of the evaluation model 
$M$.
The values in $\Delta m$ are then sorted, and samples with larger differences are preferred for subsequent auditing.
This strategy is based on the intuition that a larger prediction gap indicates a more significant behavioral difference between a model trained on the dataset and one that is not, providing a reliable basis for dataset auditing.
Here, we choose to select samples with values in the first $r_c$ proportion in $\Delta m$ as the candidate set $\mathbb{C}$,
the remaining samples are divided into the remaining set $\mathbb{E}$.

Next, we randomly select $r_r \cdot |O|$ samples from the candidate set $\mathbb{C}$ as the reference set $\mathbb{R}$,
and select another $r_m \cdot |O|$ samples as the modification set $\mathbb{M}$.
The reference and modification sets will be utilized for the auditing in the third phase.
Then, the other $(r_c-r_m-r_r)\cdot|O|$ samples in the candidate set $\mathbb{C}$ are assigned to the remaining set $\mathbb{E}$.

According to the original dataset $O$ and modified dataset $Q$,
along with the modification, reference, and remaining sets,
we construct two distinct datasets.
The first is the public dataset 
$D$, which is released and contains the modified samples from the modification set, as well as the original samples from the reference set and the remaining set. 
The second is the private dataset 
$U$, which is kept unpublished and includes the original versions of the modification set, along with the original samples from the reference and remaining sets.

\subsection{Copyright Verification}
\label{subsec:copyright_verify}

\begin{algorithm}[!t]

        \caption{Copyright Verification}
        \label{algorithm:copyright_verify}
        \KwIn{The published dataset $D$, the unpublished dataset $U$, the modification set $\mathbb{M}$, the reference set $\mathbb{R}$, 
        the suspect model $S$,
        the upper limit of the difference threshold $H$,
        the boundary of low probability $B$,
        the boundary coefficient $\beta$,
        significant level $\alpha$
        }
        \KwOut{Auditing result}
        {
        Initialized $\Delta s_R=\{\},\Delta s_M=\{\}$ \\
        // Model query \\
        \For{each sample pair (o, o') in the reference set}
        {
        Calculate the output probability of $o$ and $o'$ on the model $S$, \ie, $P_{o}^S$ and $P_{o'}^S$ \\
        Add $\Delta s_o=P_{o}^S-P_{o'}^S$ to $\Delta s_R$
        }
        // Threshold estimation \\
        $\bar h\leftarrow mean(\Delta s_R)$
        \\
        // Range constraint \\
        $h \leftarrow clip(\bar h,-H,H)$ \label{line:h_clip}
        \\
        // Model query \\
        \For{each sample pair (g', g) in modification set}
        {
        Calculate the output probability of $g'$ and $g$ on the model $S$, \ie, $P_{g'}^S$ and $P_{g}^S$ \\
        $\Delta s_g=P_{g}^S-P_{g'}^S$ \\
        // Post processing \\
        \If{$P_{g'}^S<B$ and $P_{g}^S<B$}
        {$\Delta s_g=(1+\beta)\bar h$}
         Add $\Delta s_g$ to $\Delta s_M$
        }
         Hypothesis testing based on difference sequence $\Delta s_M$ and threshold $h$ with significant level $\alpha$ \\

        }
       
\end{algorithm}

In this phase, we need to audit whether the published dataset was misused.
Considering the complexity and diversity of target models and data processing pipelines, we assess dataset usage by evaluating the performance differences between the reference set and the modification set on the target model.
The detailed process is shown in~\autoref{algorithm:copyright_verify}.

First,
for each sample pair $(o,o')$ consist of the original sample $o$ and modified sample $o'$ in the reference set $\mathbb{R}$,
we can obtain their predictions on the suspect model $S$, \ie, $P_o^S$ and $P_{o'}^S$.
Then, the prediction difference $\Delta s_o$ between
$o$ and $o'$ can be 
calculated.
After traversing and querying the samples in %
$\mathbb{R}$,
we can get 
a sequence of difference values, \ie, $\Delta s_R$.
The purpose of 
these
steps is to obtain a threshold $h$.

The core design intuition is that if the suspect model $S$ was not trained on the published dataset $D$,
its behavior on the reference set and the modification set should be similar.
Here, we define model behavior as the prediction difference between the original dataset $O$ and the modified dataset $Q$.
Conversely, if the suspect model was trained on $D$, 
we expect a noticeable behavioral discrepancy between the reference and modification sets.
Specifically, 
the prediction differences for the modification set are desired to be significantly smaller than those for the reference set.
This is because the published dataset $D$ includes the modified samples from the modification set and the original samples from the reference set.

To make this distinction operational, we compute the mean prediction difference $\bar h$ across all sample pairs in the reference set and adopt it as the decision threshold. 
However, we observe that the prediction differences within the reference set may exhibit particularly high values, leading to an overly large threshold. Consequently, a higher risk of false positives is prone to occur. 
To address this, we impose an upper limit $H$ on the difference threshold to improve the reliability of auditing and mitigate misclassification.
Regarding the value of upper limit $H$,
we provide a detailed analysis in~\autoref{subsec:appendix_threshold_analysis}.

After computing the threshold $h$ using the reference set, we query the suspect model with the sample pairs (\eg, $g'$ and $g$) from the modified set and obtain their corresponding prediction differences, denoted as $\Delta s_g$.
We then apply a post-processing step to refine $\Delta s_g$.
Specifically, if the predicted probabilities for both the original and modified samples are very low (\ie, below the predefined threshold $B=1/n_c$, where $n_c$ is the number of output categories),
we adjust $\Delta s_g$ to $(1+\beta)\bar h$, where $\beta$ is a small positive constant (\eg, 0.01).
Here, $\beta$ is utilized to guide the direction of subsequent hypothesis testing and avoid zero-difference exclusion.

The rationale behind this adjustment is: When both predicted probabilities are negligible, the suspect model $S$ is likely not to utilize the dataset $D$. 
However, due to the low absolute values, the resulting difference $\Delta s_g$
may also be small, even smaller than the threshold $h$.
This can mislead the auditor to determine that the dataset $D$ was used.
By adjusting $\Delta s_g$ slightly above the average difference in the reference set, 
we can effectively mitigate the risk of such false positives and improve the robustness of auditing.

Finally, we employ hypothesis testing to verify whether the suspect model 
$S$ exhibits dataset misuse. 
The goal of the hypothesis test is to verify whether the difference in the probability of sample pairs in the candidate set $\Delta s_M$ is significantly lower than a specific threshold $h$ (calculated from the difference in the probability of sample pairs in the reference set, \autoref{line:h_clip} in \autoref{algorithm:copyright_verify}).
If the probability difference is indeed significantly lower than the threshold (\ie, $\Delta s_M<h$), it means that the suspect model is likely trained using the protected dataset (\ie, the dataset is misused).
Conversely, if the probability difference is not significantly lower than the threshold, it determines that the dataset is not misused.
Specifically, we adopt the Wilcoxon Signed-Rank Test~\cite{woolson2005wilcoxon}, a non-parametric statistical method designed for paired samples, to implement the hypothesis testing.
Unlike parametric tests, the Wilcoxon Signed-Rank Test does not rely on the assumption of normality, making it well-suited for scenarios with small sample sizes or unknown data distributions. Moreover, this method effectively leverages both the sign and rank information of the sample differences, offering higher statistical efficiency and better interpretability. 
This makes it particularly appropriate for detecting subtle yet consistent behavioral deviations in the model.
The null hypothesis $H_0$ and alternative hypothesis $H_1$ are as follows:

\begin{equation*}
\begin{aligned}
&H_0: \Delta s_M > h,  \quad \text{(no significant misuse detected)} \\
&H_1: \Delta s_M < h,  \quad \text{(evidence of dataset misuse)}.
\end{aligned}
\end{equation*}

The testing procedure is described below:
\begin{enumerate}
    \item For each paired observation, compute the signed difference relative to the reference threshold:
    \[
    d_i = \Delta s_M^{(i)} -h, \quad i = 1, 2, \dots, n,
    \]
    where $\Delta s_M^{(i)}$ denotes the $i$-th observed difference.
    \item Exclude all instances where $d_i = 0$ (\ie, no difference).
    \item For the remaining non-zero differences, compute the absolute values $|d_i|$ and assign ranks $R_i$ based on these absolute values (average ranks in case of ties).
    \item Assign a sign to each rank based on the sign of $d_i$, and compute the sum of negative ranks:
    \[
    W = \sum_{i: d_i < 0} R_i.
    \]
    The statistic $W$ represents the total rank magnitude supporting $H_1$ (\ie, cases where $\Delta s_M^{(i)}<h$).
    \item Use the statistic $W$ to compute a one-sided $p$-value. 
    If $p < \alpha$ (\eg, $\alpha = 0.01$), reject the null hypothesis $H_0$ and conclude that the dataset $D$ was misused.
\end{enumerate}

\subsection{Putting Things Together}
\label{subsec:method_summary}

\begin{algorithm}[!t]

        \caption{\method}
        \label{algorithm:method_summary}
        \KwIn{Original dataset $O$, perturbation budget $\varepsilon$, the parameter set for three-dimensional Perlin noise $\delta_{\text{per}}=\{\lambda_x,\lambda_y,\lambda_t,\phi_{\text{sine}},\Omega\}$, the ratio of candidate set $r_c$, the ratio of modification set $r_m$, the ratio of reference set $r_r$, 
         the suspect model $S$,
        the upper limit of the difference threshold $H$,
        the boundary of low probability $B$,
        boundary coefficient $\beta$,
        significant level $\alpha$
        }
        \KwOut{Audit result}

        // Phase 1: Sample modification \\
        $M,Q,P_O^M,P_Q^M\leftarrow$ \autoref{algorithm:sample_modify}$(O,\varepsilon,\delta_{\text{per}})$\\

        // Phase 2: Sample selection \\
            $D,U,\mathbb{R},\mathbb{M},\mathbb{E}\leftarrow$ \autoref{algorithm:sample_select}$(O,Q,P^M,r)$\\

        // Phase 3: Copyright verification \\
       Audit result $\leftarrow$\autoref{algorithm:copyright_verify}$(D,U,\mathbb{R},\mathbb{M},H,B,\beta,\alpha)$\\
\end{algorithm}

The aforementioned three phases form the overall process of~\method. 
\autoref{algorithm:method_summary} describes the overall workflow of~\method.
In the first phase, we can train an evaluation model $M$ and generate a modified dataset $Q$ based on the original dataset $O$, perturbation budget $\varepsilon$ and noise parameters $\delta_{\text{per}}$.
Then, in the second phase, we can select the suitable samples that are likely to achieve the amplification effect according to the posterior probability of samples in $O$ and $Q$ on the evaluation model $M$.
The selected modified samples and other original samples are released as the public dataset $D$.
In the copyright verification phase,
we calculate the suspect model's behavior difference on the reference set and the modification set,
and apply hypothesis testing to verify whether the published dataset was used.

Note that \method can be extended to Top-$K$ or Label-only scenarios.
Specifically, different weights are assigned to 
$K$ publication categories with different rankings to approximate the probability (\eg, the weight is the inverse of the ranking). 
Then, our current scheme can be directly applied for auditing.

\section{Evaluation}
\label{sec:evaluation}

In this section, we first describe the experimental setup in~\autoref{subsec:experimental_setup},
and evaluate the overall auditing performance in~\autoref{subsec:end_to_end}.
Then, we perform ablation experiments to explore the influence of various components of~\method in~\autoref{subsec:ablation}. 
Moreover, we explore the auditing performance of various methods under Top-$K$ and Label-only settings in~\autoref{subsec:topk_label_only}.
Next, we verify the impact of the parameter setting of~\method in~\autoref{subsec:parameter_variation}.
Furthermore, we explore the robustness of~\method in~\autoref{subsec:robustness}.
Due to space constraints,
we defer the auditing results on the SSv2 dataset to~\autoref{subsec:result_larger_dataset_backbone}.
We also provide the robustness evaluation under common perturbations in~\autoref{subsec:robustness_common_pert} and auditing performance against adaptive attackers in~\autoref{subsec:robustness_adaptive_attack}.
In addition, we further provide the efficiency analysis in~\autoref{subsec:efficiency_analysis}.

\subsection{Experimental Setup}
\label{subsec:experimental_setup}

\mypara{Datasets}
We evaluate the performance of different methods on three standard benchmark datasets used: HMDB-51~\cite{kuehne2011hmdb}, UCF-101~\cite{soomro2012ucf101}, and SSv2~\cite{goyal2017something}. 
The first dataset contains $51$ categories and a total of nearly $7,000$ videos,
and the second dataset consists of more than $13,000$ videos in $101$ categories.
The last dataset contains more than $160,000$ videos,
which consists of $174$ categories.

These datasets all belong to action recognition,
which is a representative task and widely considered in existing video recognition research. 
The action recognition requires modeling both spatial and temporal information in videos, which allows us to effectively evaluate the performance of the proposed method. 
The core idea of \method is to determine auditing decisions based on behavioral differences of the suspect model in different samples. 
This method design has the potential for generalization since the model’s behavioral differences can be characterized in other tasks.

\mypara{Suspect Models}
We verify the effectiveness of our method on four typical models, \ie, I3D~\cite{carreira2017quo}, SlowFast~\cite{feichtenhofer2019slowfast}, TSM~\cite{lin2019tsm}, and TimeSformer~\cite{bertasius2021space}.
These models are also widely adopted in existing video studies~\cite{chen2021deep,al2024look}.

The I3D and TSM models utilize ResNet-50 pre-trained on the ImageNet dataset to initialize their backbones, 
and the TimeSformer model adopts the pretrained ViT model to initialize the backbone.
The SlowFast model employs a randomly initialized ResNet-50 backbone without pre-training.
In this work, we utilize open-sourced  MMAction2~\cite{2020mmaction2} to implement the above video recognition models.

\mypara{Baselines}
To our knowledge, there are no other dataset auditing designed for video recognition. 
Here, we apply the SOTA auditing methods (\ie, \mlda~\cite{huang2024general} and \mt~\cite{chen2025MembershipTracker}) for image data to the video domain as the baselines.
A detailed description and discussion can be found in~\autoref{subsec:appendix_discussion}.

\mypara{Metrics}
We adopt the following metrics to evaluate the auditing performance:
$\Delta$acc, true positive rate (TPR),
F1 score and false positive rate (FPR).
$\Delta$acc represents the difference between the test accuracy of the model trained using the modified dataset and the test accuracy of the model using the original dataset.
TPR measures the proportion of correctly identified positive samples (\ie, trained on the target dataset) among all actual positives.
F1 score balances precision and recall, providing a single harmonic mean that is particularly informative under class imbalance.
FPR quantifies the proportion of negative samples (\ie, trained on other datasets) incorrectly classified as positive.
For FPR, lower is better.
For the other three metrics, higher is better.

\mypara{Experimental Settings}
For \method, we set the perturbation budget $\varepsilon=10$ (\ie, $\ell_{\infty}$-norm constraint),
the ratio of modification samples $r_m=1\%$,
the upper limit of the difference threshold $H=0.05$, and the significant level $\alpha=0.01$.
We utilize the I3D model as the evaluation model by default.
For the parameter set of Perlin noise,
we set $\lambda_{x}=\lambda_y=32$, $\lambda_t=6.4$, $\phi_{sine}=1$, and $\Omega=2$.
The experiments are conducted with a server with 64-core EPYC AMD CPU and four NVIDIA A6000 GPUs.

\mypara{Setup}
We divide the dataset into two subsets: One designated as the protected dataset (\ie, positive), and the other as the unprotected dataset (\ie, negative). 
For the protected subset, we generate ten different published versions (\ie, $D_1,...,D_{10}$) and train ten corresponding suspect models (\ie, $S_{p1},...,S_{p10}$) for verification.
For the unprotected subset, we first create ten different datasets through random sampling and train ten models (\ie, $S_{n1},...,S_{n10}$).

During the verification phase, each published dataset $D_i$ ($i\in\{1,...,10\}$) is used to audit its corresponding suspect model $S_{pi}$,
as well as the ten models trained on unprotected data $S_{n1}$ through $S_{n10}$.
In total,
across $D_1$ to $D_{10}$,
we obtain 10 positive samples and 100 negative samples for evaluation.

\subsection{Overall Auditing Performance}
\label{subsec:end_to_end}

\begin{figure*}[!t]
    \centering
    \includegraphics[width=0.85\textwidth]{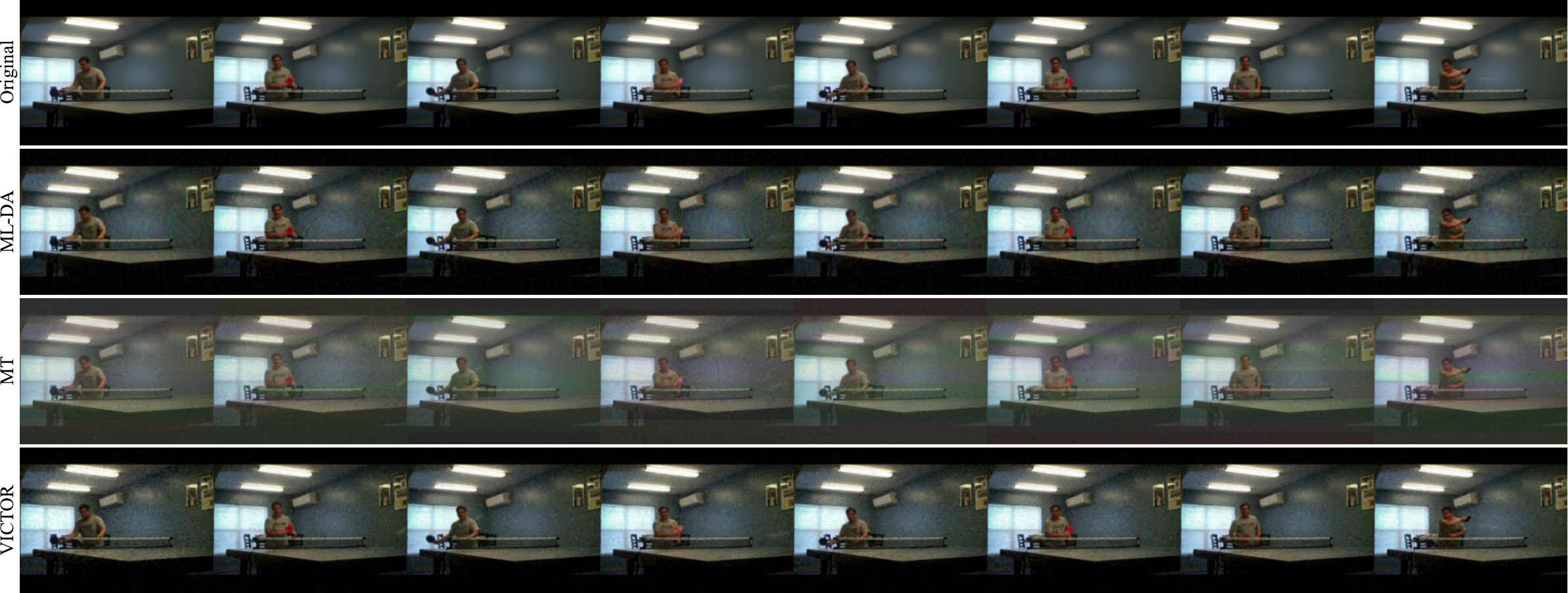}
    \caption{
    An illustration of generated videos of different methods.
    }
    \label{fig:end2end_video_show}
\vspace{-0.1cm}
\end{figure*}

\begin{table}[!t]
    \centering
    \caption{SSIM comparison 
    on various datasets.
    }
    \label{table:ssim_comparison}
    \vspace{-0.1cm}
    \footnotesize
    \setlength{\tabcolsep}{1.2em}
	\begin{tabular}{c| c | c | c   }
		\toprule
		\textbf{Dataset} & ML-DA & MT & \method   \\
		\midrule
        HMDB-51 & 0.806 & 0.681 & \textbf{0.813}  \\
        UCF-101 & 0.785 & 0.680 & \textbf{0.798}  \\
        \bottomrule
	\end{tabular}
    \vspace{-0.2cm}
\end{table}

\begin{table*}[!t]
\renewcommand{\arraystretch}{1.2}
  \centering
  \caption{Overall auditing performance on the evaluation metrics.}
  \vspace{-0.1cm}
  \scalebox{0.95}{
  \begin{tabular}{c|c|ccc|ccc|ccc}
  \hline
  \multirow{3}{*}{\textbf{Dataset}} & \textbf{Model} & \multicolumn{3}{c|}{\textbf{I3D}} & \multicolumn{3}{c|}{\textbf{SlowFast}} & \multicolumn{3}{c}{\textbf{TSM}} \\ \cline{2-11} 
                                   & \diagbox[dir=NW]{\textbf{Metric}}{\textbf{Method}} 
                                   & \textbf{\mlda~\cite{huang2024general}} & \textbf{\mt~\cite{chen2025MembershipTracker}} & \textbf{\method} 
                                   & \textbf{\mlda~\cite{huang2024general}} & \textbf{\mt~\cite{chen2025MembershipTracker}} & \textbf{\method} 
                                   & \textbf{\mlda~\cite{huang2024general}} & \textbf{\mt~\cite{chen2025MembershipTracker}} & \textbf{\method} \\ \hline
  \multirow{4}{*}{\textbf{HMDB-51}} 
    & $\Delta$acc  & -0.063 & -0.065 & \textbf{-0.020} & -0.072 & -0.052 & \textbf{-0.042} & -0.043 & -0.044 & \textbf{-0.032} \\
    & TPR          & 0.400  & \textbf{1.000}  & \textbf{1.000}  & 0.000  & \textbf{1.000}  & \textbf{1.000}  & 0.000  & \textbf{1.000}  & \textbf{1.000} \\
    & F1 Score     & 0.308  & 0.500  & \textbf{1.000}  & N/A    & 0.294    & \textbf{1.000}  & N/A    & 0.357    & \textbf{1.000} \\
    & FPR          & 0.120  & 0.200  & \textbf{0.000}  & \textbf{0.000} & 0.480 & \textbf{0.000} & 0.100 & 0.360 & \textbf{0.000} \\ \hline
  \multirow{4}{*}{\textbf{UCF-101}} 
    & $\Delta$acc  & -0.024 & -0.038 & \textbf{-0.021} & -0.014 & -0.036 & \textbf{-0.012} & -0.015 & -0.038 & \textbf{-0.011} \\
    & TPR          & 0.000  & 0.700  & \textbf{1.000}  & 0.000  & 0.800  & \textbf{1.000}  & 0.000  & 0.800  & \textbf{1.000} \\
    & F1 Score     & N/A    & 0.609   & \textbf{1.000}  & N/A    & 0.229    & \textbf{1.000}  & N/A    & 0.889    & \textbf{1.000} \\
    & FPR          & 0.040  & 0.060  & \textbf{0.000}  & 0.210  & 0.520  & \textbf{0.000}  & \textbf{0.000} & \textbf{0.000} & \textbf{0.000} \\ \hline
  \end{tabular}}
  \vspace{-0.25cm}
  \label{table:audit-performance}
\end{table*}

In the section, we explore the overall auditing performance of \method and the baseline.
First, we provide an illustration of generated videos of various methods, as shown in~\autoref{fig:end2end_video_show}.
We further provide the average SSIM comparison of different methods on various datasets, as illustrated in~\autoref{table:ssim_comparison}.
We observe that the changes introduced by~\method are relatively smooth across video frames, whereas \mlda exhibits more abrupt transitions between frames due to the independent nature of its perturbations.
\mt suffers from the highest distortion because it mixes the original image with other images.
The results in~\autoref{table:ssim_comparison} are consistent with our analysis above.
The SSIM of \method is the highest, followed by \mlda, while \mt is significantly lower than both.

\autoref{table:audit-performance} provides the overall auditing performance on the four evaluation metrics of various methods.
We make the following observations.
First, \method consistently outperforms the baselines across all evaluation metrics, demonstrating its effectiveness. 
This is attributed to its ability to amplify the influence of modified samples on the target model, thereby enabling more accurate detection of dataset misuse.
Second, the value of $\Delta \text{acc}$ on the HMDB-51 dataset is lower than that on UCF-101. 
This is likely due to the smaller scale of HMDB-51, where the impact of  modified samples is more pronounced.
Third, for \mlda, the TPR is zero in multiple scenarios, and F1 scores cannot even be computed because both precision and recall are zero.
On the one hand,
\mlda relies on comparing the predictions of published and unpublished data. 
This is less effective in the video domain, where models make decisions based on the joint information from multiple frames.
On the other hand,
\mlda requires a relatively high proportion of modified samples to be effective. 
When only a small number of samples are modified, \mlda struggles to produce meaningful results.
Moreover, for \mt,
the TPR reaches 1 on HMDB-51, but the corresponding FPR remains high. 
This is because \mt determines the decision threshold based on the loss distribution of non-member data on the suspect model. 
However, the loss of member data on a model that has not actually been trained on the target dataset can still fall below this threshold, leading to frequent misclassification. 
The issue is further exacerbated when the suspect model is trained from scratch (\ie, SlowFast), as its output tends to be more random. 
Furthermore, the accuracy of \mt drops significantly due to its heavy distortion of the original data.
In addition, suspect models trained on larger datasets (\ie, UCF-101) generally exhibit stronger generalization, which reduces the loss gap between member and non-member samples and leads to a decline in TPR.
In contrast,
\method achieve $100\%$ accuracy on multiple datasets and models, which emphasizes the practicality of~\method.

\subsection{Ablation Study}
\label{subsec:ablation}

In this section, we conduct ablation studies to evaluate the influence of different components. 
Specifically, we examine the effectiveness of the evaluation model, threshold clipping, and post processing on the overall auditing performance.

\mypara{Effectiveness of Evaluation Model}
Recalling~\autoref{subsec:sample_modification} and~\autoref{subsec:sample_select},
we adopt an evaluation model to verify the effectiveness of modified samples, which can help the selection of samples.
In this section, our aim is to explore the necessity of the evaluation model.
\autoref{fig:ablation_eval_model_ucf101} illustrates the TPR and FPR results with and without the evaluation model.
Here, ``Without Evaluation Model'' means that the modified and reference samples are randomly selected.

We observe that the verification accuracy is slightly lower when the evaluation model is removed, compared to the case where it is applied. 
Specifically, the TPR decreases when the suspect model is I3D, and the FPR becomes greater than zero when the suspect model is SlowFast. 
This is because, without guidance from the evaluation model, the selected modified and reference samples may have less consistent effects on the target model, leading to a slight decline in auditing performance. 
The above results indicate that the evaluation model is helpful to enhance auditing effectiveness.
At the same time, we would like to clarify that \method can still achieve competitive auditing performance without the evaluation model, rather than failing completely.

\begin{figure}[!t]
    \centering
    \includegraphics[width=0.35\textwidth]{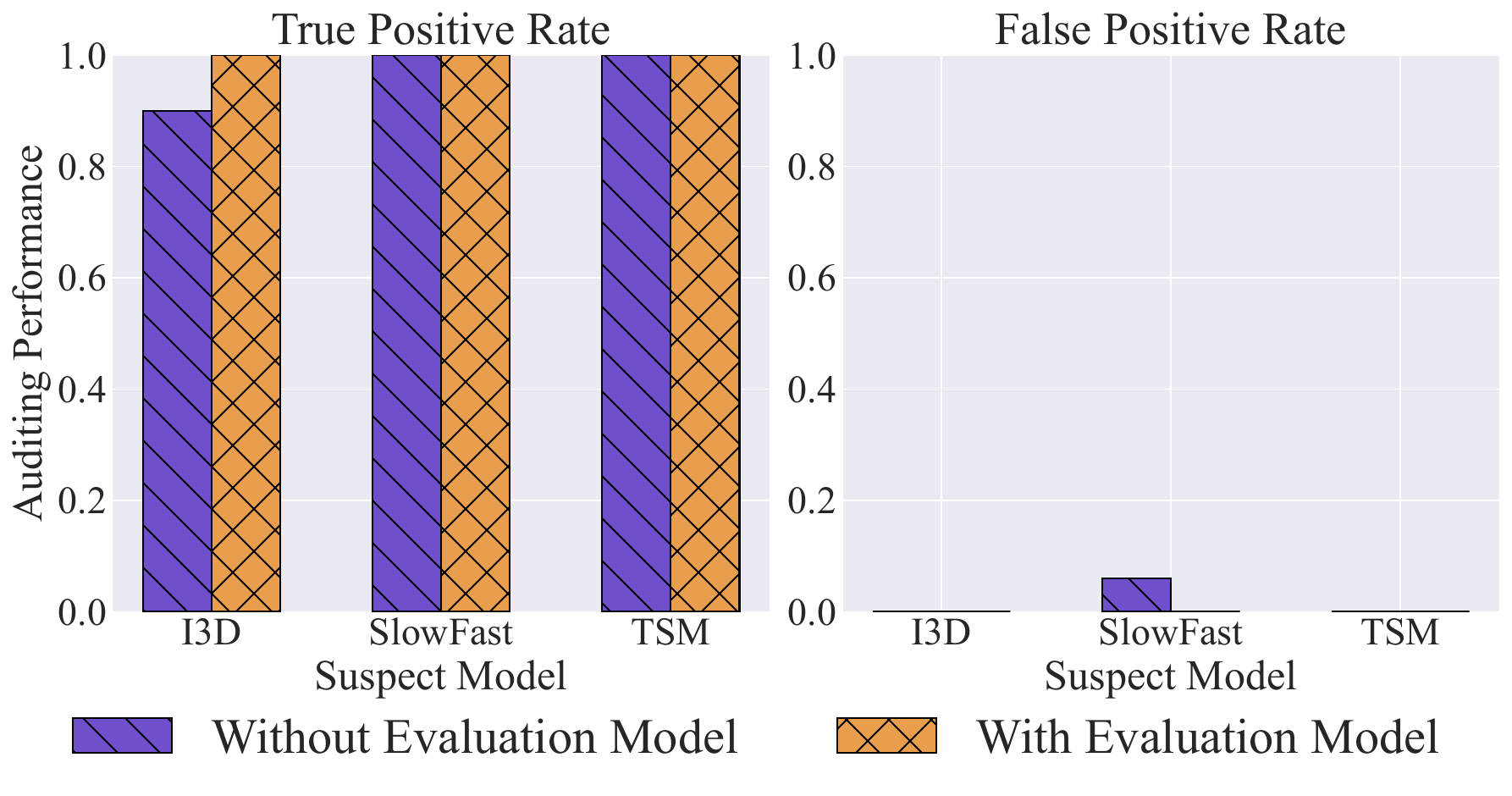}
    \vspace{-0.3cm}
    \caption{
    The effectiveness of evaluation model on the two metrics of the UCF-101 dataset.
    }
    \label{fig:ablation_eval_model_ucf101}
\vspace{-0.5cm}
\end{figure}

\mypara{Effectiveness of Threshold Clipping}
Recalling~\autoref{subsec:copyright_verify},
we introduce a threshold clipping mechanism to mitigate potential false positives. 
Here,
we examine the auditing performance with and without this mechanism. 
\autoref{fig:ablation_clip_ucf101} presents the TPR and FPR on the UCF-101 dataset under both settings. The results show that while the TPR remains consistent between the two methods, the FPR is substantially higher when threshold clipping is disabled. 
This demonstrates that threshold clipping effectively reduces the false positive rate of \method. Moreover, we observe that the impact of threshold clipping is more pronounced for the I3D and TSM models. 
This is primarily because these models are fine-tuned from pretrained networks, which endows them with stronger generalization capabilities.
In this case, their output differences between modified and original samples are smaller than SlowFast.

\begin{figure}[!t]
    \centering
    \includegraphics[width=0.35\textwidth]{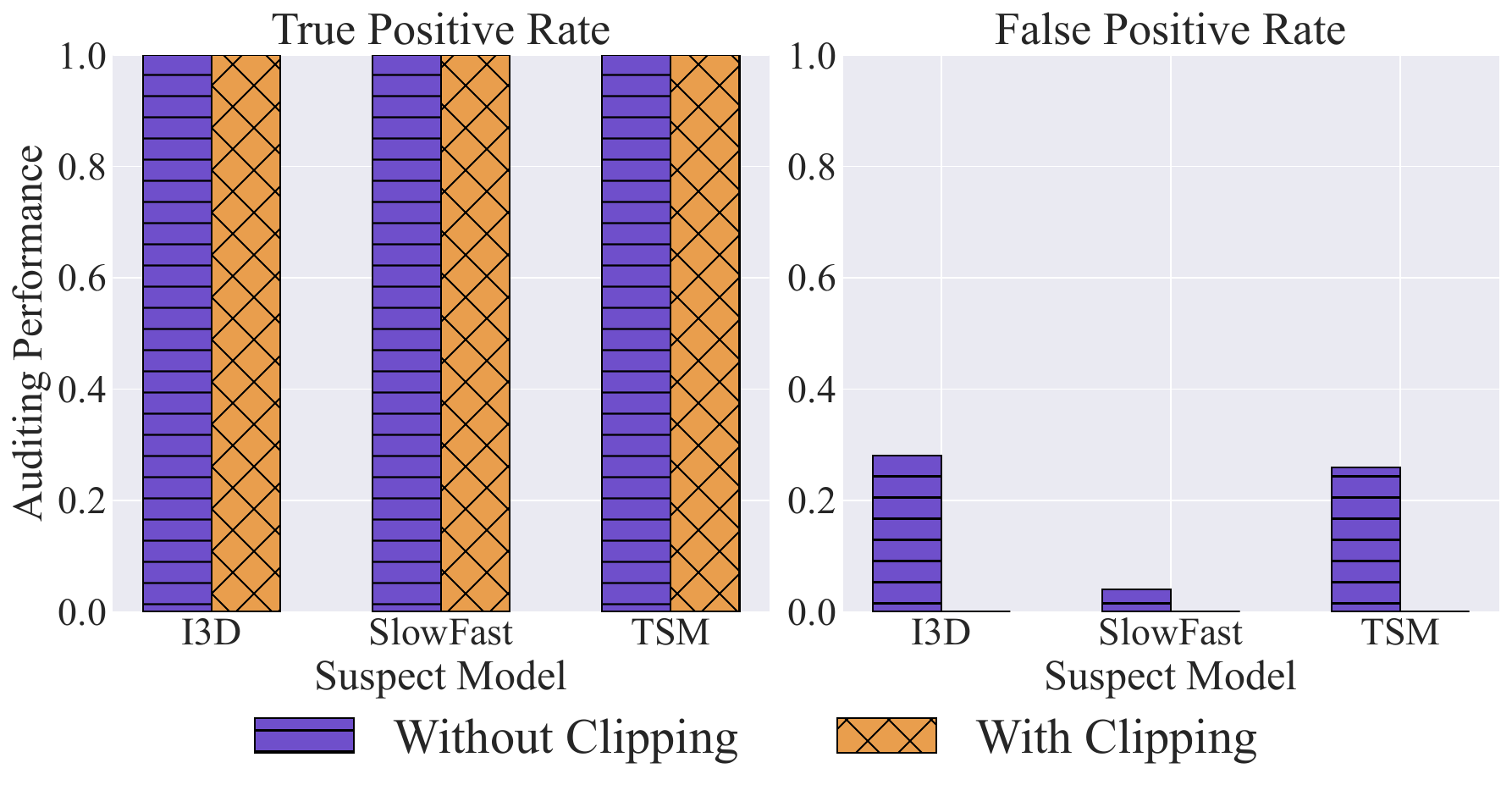}
    \vspace{-0.2cm}
    \caption{
    The effectiveness of threshold clipping on the two metrics of the UCF-101 dataset.
    }
    \label{fig:ablation_clip_ucf101}
\vspace{-0.35cm}
\end{figure}

\mypara{Effectiveness of Post Processing}
During the copyright verification phase, we apply post-processing to the output probability differences. 
In this section, we investigate the necessity of this post-processing step. \autoref{fig:ablation_pp_hmdb51} presents the auditing results on the HMDB-51 dataset with and without post-processing. 
Without post-processing, various models exhibit false positives. 
This is mainly because some suspect models have relatively weak predictive capabilities, leading to low output probabilities for both the unused original samples and their modified counterparts. Consequently, the output differences are also small, which increases the likelihood of misjudgment.
In contrast, after applying post-processing, the FPR drops to zero. 
These results emphasize the key role of post-processing in improving the auditing accuracy of \method.

\begin{figure}[!t]
    \centering
    \includegraphics[width=0.35\textwidth]{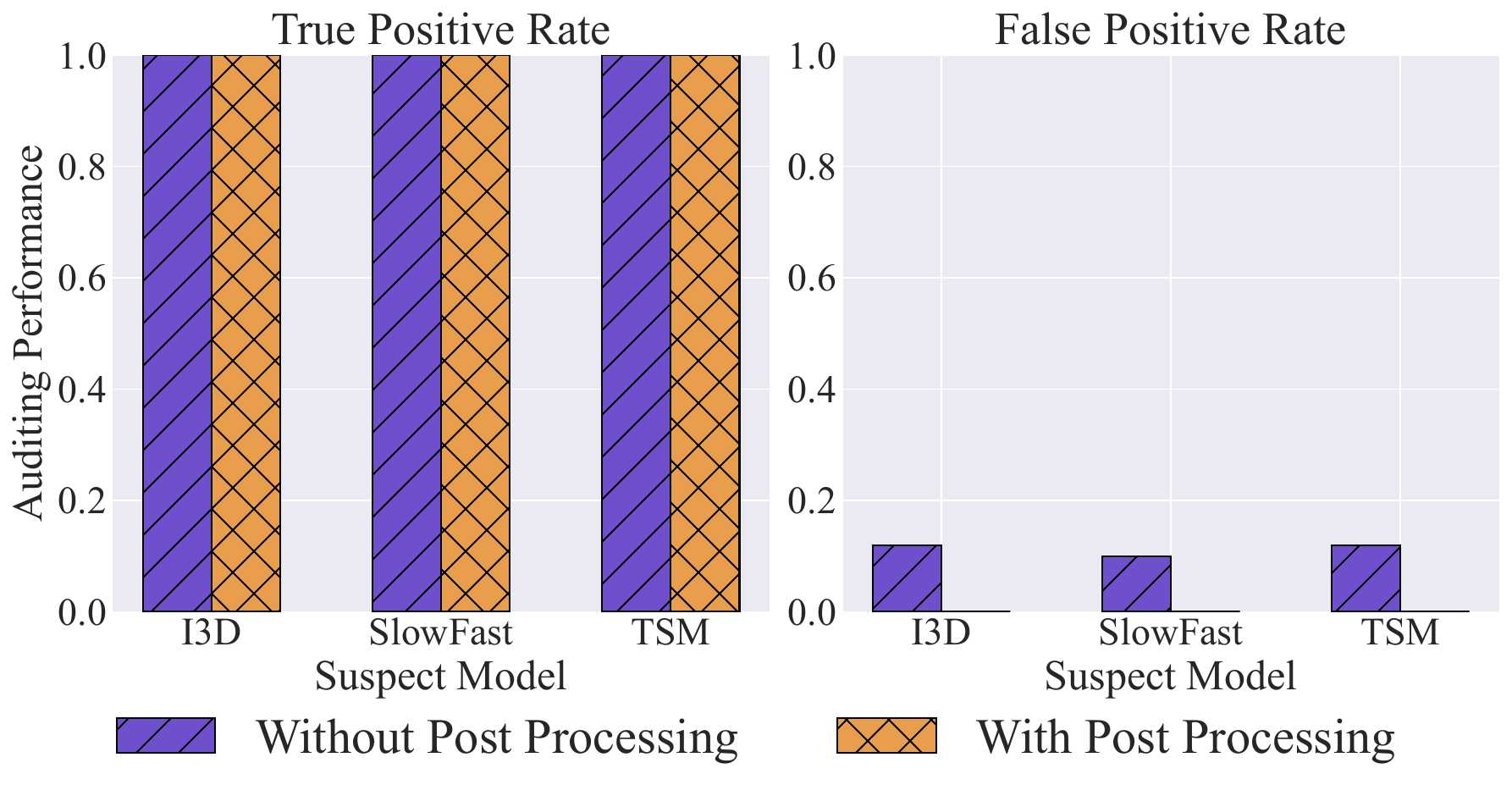}
    \vspace{-0.2cm}
    \caption{
    The effectiveness of post processing on the two metrics of the HMDB-51 dataset.
    }
    \label{fig:ablation_pp_hmdb51}
\vspace{-0.5cm}
\end{figure}

\subsection{Top-$K$ and Label-only Settings}
\label{subsec:topk_label_only}

In this section, we consider two more challenging scenarios (\ie, Top-$K$ and Label-only) to validate the effectiveness of \method.
In the Top-$K$ ($K=5$) setting,
only the labels of the top $K$ highest-probability samples are provided.
In the Label-only setting,
only the label corresponding to the highest-probability sample is output.

\autoref{fig:topk_label_hmdb51} and \autoref{fig:topk_label_ucf101}
illustrate the auditing performance on the HMDB-51 and UCF-101 datasets.
For the Top-$K$ scenario,
\method still exhibits great auditing performance across multiple suspect models and datasets (\ie, TPR = 1 and FPR = 0).
This highlights the versatility and significant advantages of \method in dataset auditing scenarios.
For the Label-only scenario,
the TPR of \method in partial cases shows a slight decrease.
The reason is that the obtained information under this scenario is limited, posing a huge challenge to accurate auditing.
Nevertheless, \method still demonstrates competitive auditing performance and low FPRs.

\begin{figure}[htbp]
    \centering
    \includegraphics[width=0.35\textwidth]{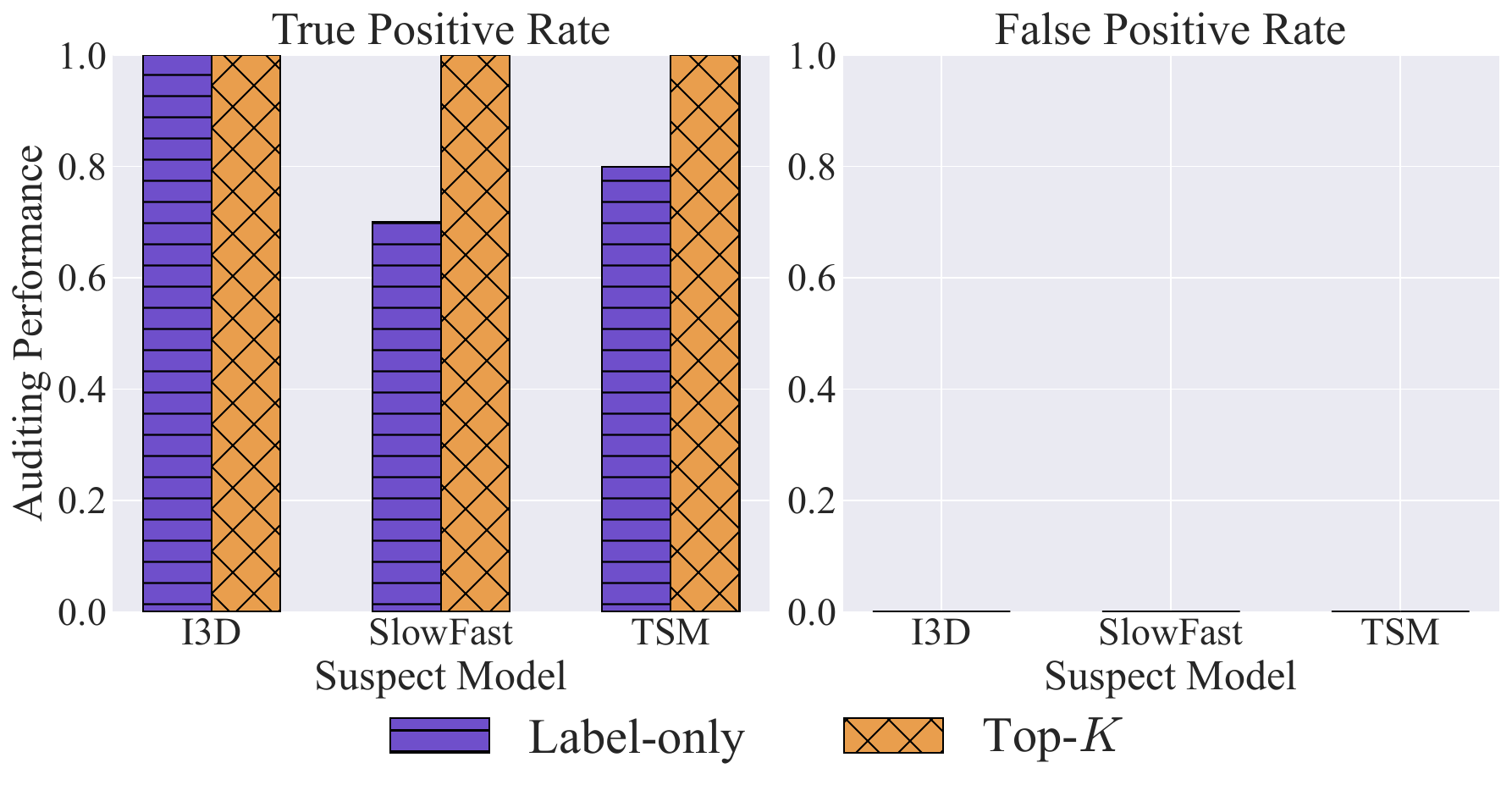}
    \vspace{-0.2cm}
    \caption{
    The auditing performance under Top-$K$ and Label-only settings of the HMDB-51 dataset.
    }
    \label{fig:topk_label_hmdb51}
\vspace{-0.5cm}
\end{figure}

\begin{figure}[htbp]
    \centering
    \includegraphics[width=0.35\textwidth]{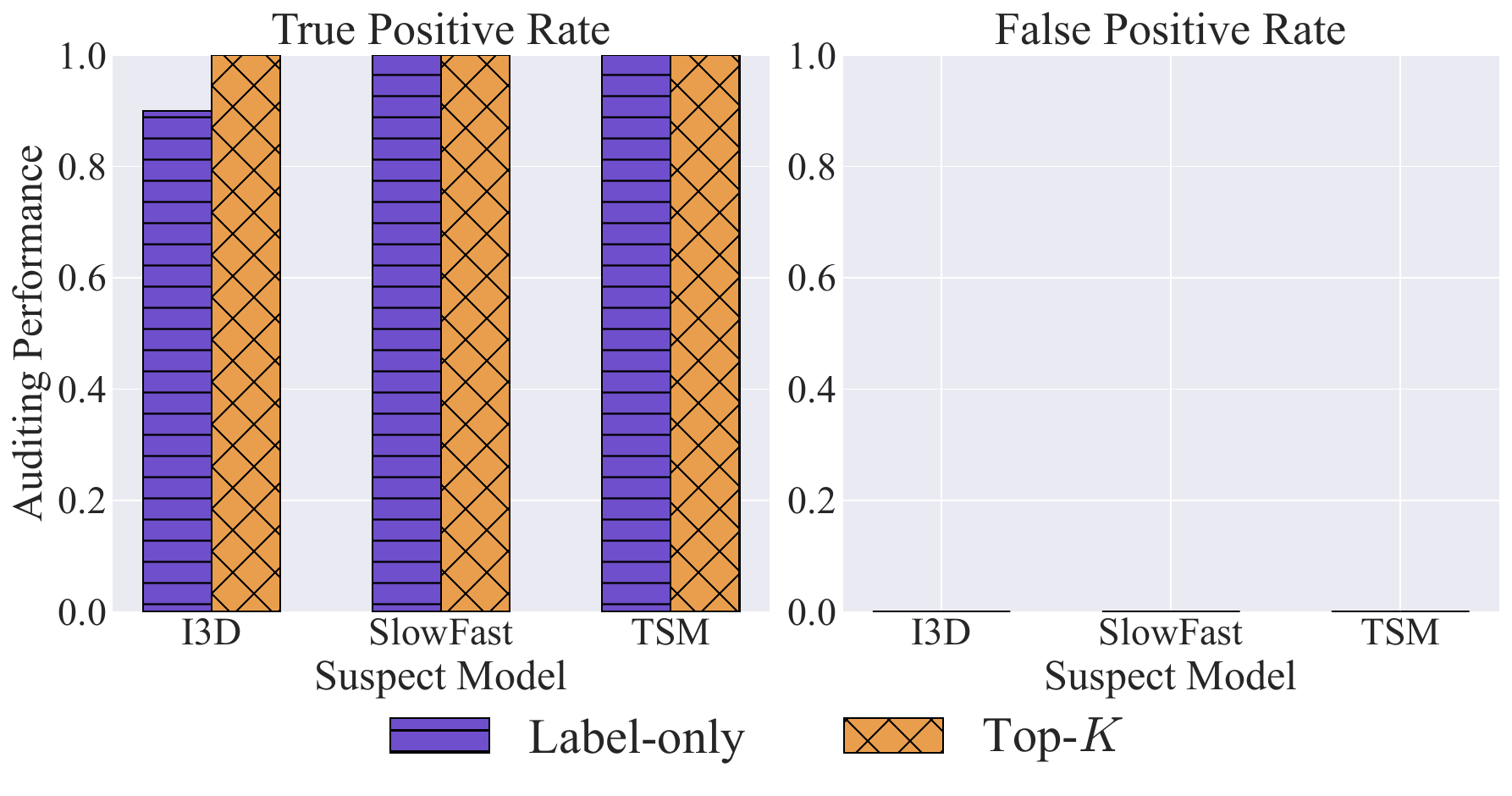}
    \vspace{-0.2cm}
    \caption{
    The auditing performance under Top-$K$ and Label-only settings of the UCF-101 dataset.
    }
    \label{fig:topk_label_ucf101}
\vspace{-0.5cm}
\end{figure}

\subsection{Parameter Variation}
\label{subsec:parameter_variation}

In this section, our aim is to analyze the impact of various parameter settings in \method on the final auditing performance.
In particular, we examine the effects of the perturbation budget and the modification ratio in the section.
The analysis of the evaluation model, the noise setting and the threshold setting is deferred to~\autoref{subsec:appendix_parameter_variation} due to space limitations.

\mypara{Impact of Perturbation Budget}
\autoref{fig:param_vary_epsilon_ucf101} illustrates the performance of~\method on four evaluation metrics on the UCF101 dataset as the perturbation budget $\varepsilon$ increases from 2 to 10. 
We observe that as the budget increases from 2 to 6, the TPR and F1 score consistently improve, the FPR decreases significantly, and $\Delta acc$ drops slightly. When the budget reaches 6, the auditing accuracy across all three models is already very high. Notably, even with a perturbation budget as low as 2, \method still achieves a TPR exceeding 0.6 and an FPR below 0.2. 
These results demonstrate that \method possesses strong auditing capabilities and can deliver impressive performance even under a limited perturbation budget.

\begin{figure*}[!t]
    \centering
    \includegraphics[width=0.85\textwidth]{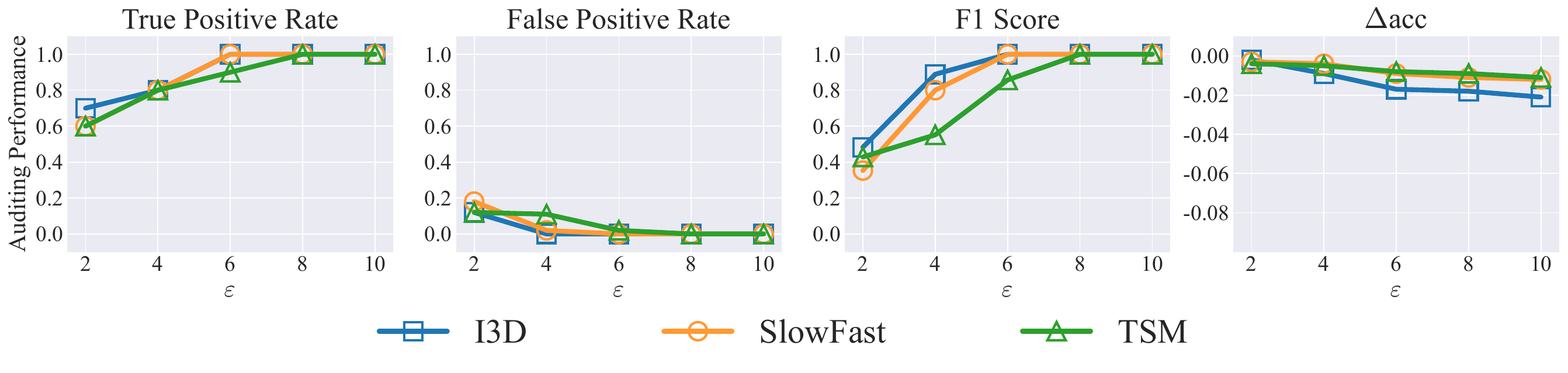}
    \vspace{-0.35cm}
    \caption{
    The impact of perturbation budget for three various suspect models on the four metrics of the UCF-101 dataset.
    }
    \label{fig:param_vary_epsilon_ucf101}
\vspace{-0.35cm}
\end{figure*}

\mypara{Impact of Modification Ratio}
\autoref{fig:param_vary_ratio_ucf101} illustrates the impact of different modification ratios on auditing performance using the UCF101 dataset. 
It can be observed that when the modification ratio is as low as 0.5\%, \method achieves an FPR of 0 across all three suspect models, along with a TPR of 100\% on the I3D and SlowFast models. 
The TPR on the TSM model in this case is 0.6, possibly because TSM is fine-tuned from a pre-trained model with relatively few fine-tuning iterations, making the influence of a low modification ratio less pronounced.
Overall, \method can achieve promising auditing performance even at a small modification ratio, showcasing its effectiveness.
When the modification ratio increases to 1\% or 2\%, \method achieves perfect auditing accuracy on all three suspect models (\ie, TPR = 1 and FPR = 0), while maintaining limited impact on normal task performance. 
These results confirm that \method can achieve high-precision dataset auditing with low modification cost.

\begin{figure*}[!t]
    \centering
    \includegraphics[width=0.85\textwidth]{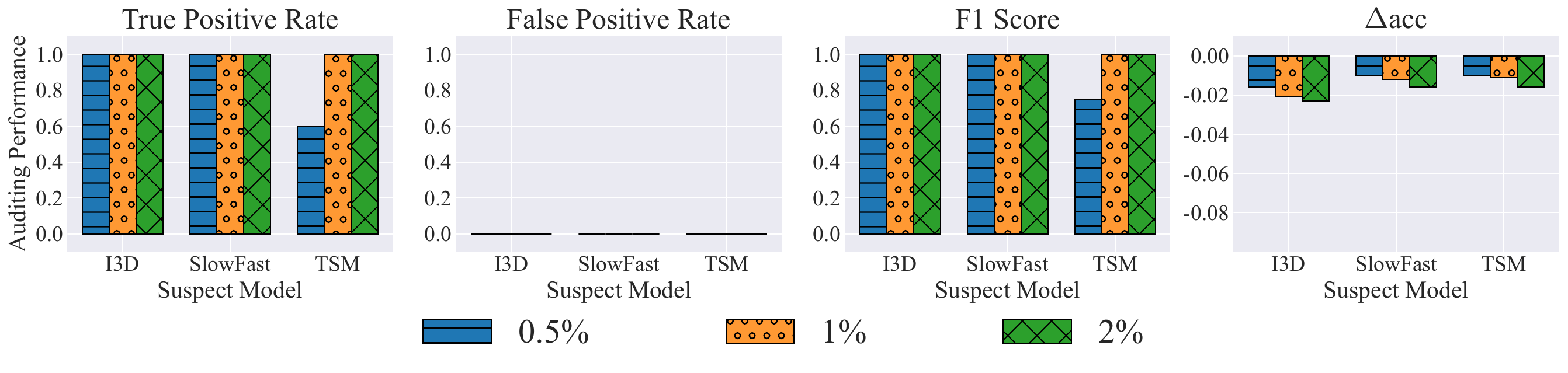}
    \vspace{-0.35cm}
    \caption{
    The impact of modification ratio for three various suspect models on the four metrics of the UCF-101 dataset.
    }
    \label{fig:param_vary_ratio_ucf101}
\vspace{-0.3cm}
\end{figure*}

\subsection{Robustness}
\label{subsec:robustness}

In this section, we investigate the robustness of \method when the pipeline of target models is perturbed to evade auditing.
Here,
we consider three types of classic mechanisms:
input preprocessing,
training intervention, post-adjustment.

\mypara{Input preprocessing}
We explore two various approaches to handle the input: input perturbation and input detection.
Regarding the input perturbation, this method tries to inject noise into each frame of all samples.
Here, the added noise is sampled from a Gaussian distribution with standard deviation $\sigma=10$.
For input detection, this approach tries to identify abnormal samples and remove them from the training. 
The details of input detection are deferred to~\autoref{subsec:appendix_robustness}.

\autoref{table:input_pert}
provides the auditing results after introducing the input perturbation. 
We summarize the following key observations:
First, adding input noise has a more pronounced negative impact on the performance of the target model in the normal task. 
Second, despite the presence of input perturbation, \method still successfully identifies all suspect models on the UCF-101 dataset. 
On the other hand, the TPR for the SlowFast and TSM models on the HMDB-51 dataset shows a slight decline. 
This may be attributed to the smaller scale of the HMDB-51 dataset, where input noise can have a stronger influence on the optimization of the target model’s parameters.
Third, \method consistently achieves an FPR of 0 across all settings. 
This highlights the robustness and reliability of \method in dataset auditing. 
In practical scenarios, falsely accusing a model of dataset misuse can lead to severe consequences, including reputational harm, financial costs, and strain of judicial resources, particularly in high-stakes intellectual property disputes.

\begin{table}[!t]
\renewcommand{\arraystretch}{1.2}
\centering
\caption{Auditing performance under input perturbation.}
\vspace{-0.15cm}
\scalebox{0.95}
{
\begin{tabular}{c|c|c|c|c}
\hline
\textbf{Dataset} & \diagbox{\textbf{Metric}}{\textbf{Model}} & \textbf{I3D} & \textbf{SlowFast} & \textbf{TSM} \\ \hline
\multirow{4}{*}{\textbf{HMDB-51}} & $\Delta$acc  & -0.053 & -0.068 & -0.052 \\
                        & TPR          & 1.000  & 0.900  & 0.900 \\
                        & F1 Score     & 1.000  & 0.947  & 0.947 \\
                        & FPR          & 0.000  & 0.000  & 0.000 \\ \hline
\multirow{4}{*}{\textbf{UCF-101}} & $\Delta$acc  & -0.060 & -0.045 & -0.028 \\
                        & TPR          & 1.000  & 1.000  & 1.000 \\
                        & F1 Score     & 1.000  & 1.000  & 1.000 \\
                        & FPR          & 0.000  & 0.000  & 0.000 \\ \hline

\end{tabular}}
\label{table:input_pert}
\end{table}

\begin{table}[!t]
\renewcommand{\arraystretch}{1.2}
\centering
\caption{Auditing performance under early stopping.}
\vspace{-0.15cm}
\scalebox{0.95}
{
\begin{tabular}{c|c|c|c|c}
\hline
\textbf{Dataset} & \diagbox{\textbf{Metric}}{\textbf{Model}} & \textbf{I3D} & \textbf{SlowFast} & \textbf{TSM} \\ \hline
\multirow{4}{*}{\textbf{HMDB-51}} & $\Delta$acc  & -0.162 & -0.361 & -0.185 \\
                        & TPR          & 0.500  & 0.300  & 0.400 \\
                        & F1 Score     & 0.667  & 0.462  & 0.571 \\
                        & FPR          & 0.000  & 0.000  & 0.000 \\ \hline
\multirow{4}{*}{\textbf{UCF-101}} & $\Delta$acc  & -0.160 & -0.388 & -0.073 \\
                        & TPR          & 1.000  &  0.600  & 1.000 \\
                        & F1 Score     & 1.000  & 0.750  & 1.000 \\
                        & FPR          & 0.000  & 0.000  & 0.000 \\ \hline

\end{tabular}}
\label{table:early_stop}
\end{table}

\mypara{Training Intervention}
In this countermeasure, a malicious attacker can train the video recognition model for a small number of epochs to prevent the model from overfitting the training samples (\ie, early stopping).
Here, we choose to train the models only for 10 epochs.

\autoref{table:early_stop} presents the auditing performance on the two datasets when early stopping is applied. 
We make the following observations:
First, the normal performance of all three target models declines significantly under early stopping, with the SlowFast model experiencing the most severe drop. 
This is primarily because SlowFast is trained from scratch without the benefit of pretrained weights, and thus requires more training epochs to learn and update the model parameters.
In this case, \method demonstrates better auditing performance on the I3D and TSM models compared to the SlowFast model. 
For example, on the UCF-101 dataset, \method achieves a TPR of 1.0 for both I3D and TSM, while the TPR for SlowFast drops to 0.6.
In addition, we find that \method shows more effective auditing on the UCF-101 dataset than on the HMDB51 dataset. 
This is likely due to the larger scale of UCF-101, which allows the target model to learn more knowledge within the same number of training rounds, thereby enabling \method to more accurately detect dataset misuse.
Actually, when early stopping is applied, the accuracy for the video models like SlowFast on the HMDB-51 dataset is very low (\ie, around $20$\%).
This indicates that the model has not fully captured the characteristics of the dataset, making the auditing effect less than ideal. 
At the same time, it should also be noted that the inherent value of this type of model is limited.
Finally, \method maintains an FPR of 0 across all datasets and models, demonstrating its reliability and practicality.

\begin{table*}[!htbp]
\renewcommand{\arraystretch}{1.2}
\centering
\caption{Auditing performance under post-adjustment.}
\vspace{-0.1cm}
\label{table:output_pert_combine}
\begin{tabular}{c|c|ccc|ccc|ccc}
\hline
\multirow{3}{*}{\textbf{Dataset}} & \textbf{Type} & \multicolumn{3}{c|}{\textbf{Fine-tuning}} & \multicolumn{3}{c|}{\textbf{Model pruning}} & \multicolumn{3}{c}{\textbf{Output noise}} \\
\cline{2-11}
 & \diagbox[width=8em]{\textbf{Metric}}{\textbf{Model}} 
 & \textbf{I3D} & \textbf{SlowFast} & \textbf{TSM} 
 & \textbf{I3D} & \textbf{SlowFast} & \textbf{TSM} 
 & \textbf{I3D} & \textbf{SlowFast} & \textbf{TSM} \\
\hline
\multirow{4}{*}{\textbf{HMDB-51}} & $\Delta$acc  & -0.032 & -0.018 & -0.047 & -0.370 & -0.396 & -0.424 & -0.181 & -0.287 & -0.222 \\
                        & TPR          & 0.800  & 1.000  & 0.900  & 0.500  & 0.100  & 0.600  & 1.000  & 0.800  & 0.900  \\
                        & F1 Score     & 0.889  & 1.000  & 0.947  & 0.667  & 0.182  & 0.750  & 1.000  & 0.889  & 0.947  \\
                        & FPR          & 0.000  & 0.000  & 0.000  & 0.000  & 0.000  & 0.000  & 0.000  & 0.000  & 0.000  \\ 
\hline
\multirow{4}{*}{\textbf{UCF-101}} & $\Delta$acc  & -0.041 & -0.002 & -0.011 & -0.593 & -0.582 & -0.616 & -0.156 & -0.185 & -0.126 \\
                        & TPR          & 1.000  & 0.900  & 1.000  & 1.000  & 0.900  & 1.000  & 1.000  & 1.000  & 1.000  \\
                        & F1 Score     & 1.000  & 0.947  & 1.000  & 1.000  & 0.947  & 1.000  & 1.000  & 1.000  & 1.000  \\
                        & FPR          & 0.000  & 0.000  & 0.000  & 0.000  & 0.000  & 0.000  & 0.000  & 0.000  & 0.000  \\ 
\hline
\end{tabular}
\end{table*}

\mypara{Post Adjustment}
\autoref{table:output_pert_combine} provides the auditing performance under various post-adjustment methods.
First, we investigate the impact of model fine-tuning as a potential evasion strategy. 
Specifically, we assume that a malicious attacker possesses a dataset with a distribution similar to the published dataset (approximately 10\% of the original dataset's scale) and uses it to fine-tune the suspect model.

The three columns corresponding to ``Fine-tuning'' report the performance of \method across four evaluation metrics on the fine-tuned models. 
First, we observe that the performance of the SlowFast model on normal tasks improves after fine-tuning. 
This is mainly because the SlowFast model is trained from scratch, and exposure to a dataset with a similar distribution can further enhance its learning.
Second, the TPR of \method shows a slight decline on some scenarios. 
This is attributed to the fine-tuning process, which diminishes the influence of the original published dataset on the target model, thereby reducing the detectability of dataset misuse.
Finally, \method consistently achieves an FPR of 0 across all settings. 
Overall, the above results and analysis suggest that \method keeps promising performance against model fine-tuning.

Next, we explore the influence of another evasion strategy (\ie, model pruning). 
The middle three columns (corresponding to ``Model pruning'') in~\autoref{table:output_pert_combine} present the auditing performance of \method on the HMDB51 and UCF-101 datasets when the pruning ratio is set to 40\%. 
We observe that the pruned models exhibit a substantial performance drop on the normal task.
Under this condition, the TPR of \method decreases significantly on the HMDB-51 dataset, while it remains high on the UCF-101 dataset. 
This discrepancy is primarily due to the smaller scale of the HMDB-51 dataset, which limits the performance of the trained model. 
After pruning, the test accuracy on HMDB-51 becomes extremely low, posing considerable challenges to the audit process. In contrast, although model performance on UCF-101 also declines after pruning, it remains superior to that of the pruned models on HMDB-51, enabling \method to maintain high auditing effectiveness.
Moreover, \method still achieves zero false positives across all three models on the two datasets. 
In summary, while model pruning can reduce detection effectiveness to some extent, this operation also severely impairs the model’s performance on normal tasks, limiting its practicality.

Finally, we investigated the impact of introducing noise into the model's output. Specifically, we injected zero-mean Gaussian noise with a standard deviation of 0.1 into the model outputs to simulate this disturbance. 
The last three columns (\ie, ``Output noise'') in~\autoref{table:output_pert_combine} present the auditing results of \method under this setting.
On the one hand, the injected noise significantly degrades the model’s performance on the original normal task. 
On the other hand, \method maintains 100\% auditing accuracy on the UCF-101 dataset, while its TPR exhibits a slight decline on the HMDB-51 dataset. 
This difference can be attributed to the varying scales of these two datasets, which lead to differences in the performance of the corresponding target models.
These results demonstrate that \method remains robust even in the presence of output noise, further validating its effectiveness in practical audit scenarios.

\section{Discussion}
\label{sec:discussion}

\method is the first auditing mechanism designed for video data that verifies whether a deployed video recognition model was trained on the specific dataset.
At the same time,
\method still has several limitations that warrant further investigation.
First, black-box query access may be unavailable for fully offline-deployed systems. 
Further, sophisticated adversaries aware of our auditing mechanism could develop adaptive evasion strategies. 
Therefore, evaluating robustness against more sophisticated adaptive evasion strategies and reducing computational overhead for scalable deployment will be valuable future work.

\section{Conclusion}
\label{sec:conclusion}

In this paper, we propose~\method, a method for verifying whether a target video dataset is used to train a suspect recognition model.
The core idea is to amplify the influence of published modified samples on the
prediction behavior of the target model.
By injecting carefully designed noise and selecting specific samples, combined with key hypothesis testing,
\method can achieve 100\% auditing accuracy across two datasets and three target models while
keeping the modifications nearly imperceptible.
Extensive ablation experiments validate the effectiveness of various components.
Furthermore, we also explore the impact of different parameter settings.
Finally, we evaluate the robustness of proposed \method.

\section*{Ethics Considerations}
This paper focuses on dataset copyright auditing in video recognition systems.
We strictly followed ethical guidelines by using publicly available, open-source datasets, under licenses that permit research and educational use.
As these datasets were curated and released by third parties, direct informed consent was not applicable. 
However, we are committed to ethical data use and will
comply with all licensing terms for any future modifications
or redistribution.
We aim to advance technological development while upholding academic ethics and data copyright norms, thereby supporting a healthy research community and fostering a culture of lawful and compliant dataset sharing.

\section*{Acknowledgments} 
We thank the anonymous reviewers for their constructive and insightful feedback. 
This work is supported 
in part 
by the National Natural Science Foundation of China under Grants No. (62402431, 62441618, 62025206, U23A20296, U24A20237, 62402379, 72594583011, 7257010373),
Key R\&D Program of Zhejiang Province under Grants No. (2024C01259, 2025C01061, 2024C01065, 2024C01012, 2025C01089),
the project CiCS of the research programme Gravitation which is (partly) financed by the Dutch Research Council (NWO) under Grant 024.006.037,
the China Postdoctoral Science Foundation under Grants No. (2025M771501, BX20250380),
and Zhejiang University.

{
    \footnotesize
    \bibliography{easy}

\begin{thebibliography}{10}
\providecommand{\url}[1]{#1}
\csname url@samestyle\endcsname
\providecommand{\newblock}{\relax}
\providecommand{\bibinfo}[2]{#2}
\providecommand{\BIBentrySTDinterwordspacing}{\spaceskip=0pt\relax}
\providecommand{\BIBentryALTinterwordstretchfactor}{4}
\providecommand{\BIBentryALTinterwordspacing}{\spaceskip=\fontdimen2\font plus
\BIBentryALTinterwordstretchfactor\fontdimen3\font minus \fontdimen4\font\relax}
\providecommand{\BIBforeignlanguage}[2]{{%
\expandafter\ifx\csname l@#1\endcsname\relax
\typeout{** WARNING: IEEEtran.bst: No hyphenation pattern has been}%
\typeout{** loaded for the language `#1'. Using the pattern for}%
\typeout{** the default language instead.}%
\else
\language=\csname l@#1\endcsname
\fi
#2}}
\providecommand{\BIBdecl}{\relax}
\BIBdecl

\bibitem{wu2022survey}
F.~Wu, Q.~Wang, J.~Bian, N.~Ding, F.~Lu, J.~Cheng, D.~Dou, and H.~Xiong, ``{A Survey on Video Action Recognition in Sports: Datasets, Methods and Applications},'' \emph{IEEE Transactions on Multimedia}, vol.~25, pp. 7943--7966, 2022.

\bibitem{deldjoo2016content}
Y.~Deldjoo, M.~Elahi, P.~Cremonesi, F.~Garzotto, P.~Piazzolla, and M.~Quadrana, ``{Content-Based Video Recommendation System Based on Stylistic Visual Features},'' \emph{Journal on Data Semantics}, 2016.

\bibitem{sun2022human}
Z.~Sun, Q.~Ke, H.~Rahmani, M.~Bennamoun, G.~Wang, and J.~Liu, ``{Human Action Recognition from Various Data Modalities: A Review},'' \emph{IEEE TPAMI}, vol.~45, no.~3, pp. 3200--3225, 2022.

\bibitem{biparva2022video}
M.~Biparva, D.~Fern{\'a}ndez-Llorca, R.~I. Gonzalo, and J.~K. Tsotsos, ``{Video Action Recognition for Lane-Change Classification and Prediction of Surrounding Vehicles},'' \emph{IEEE Transactions on Intelligent Vehicles}, vol.~7, no.~3, pp. 569--578, 2022.

\bibitem{elharrouss2021review}
O.~Elharrouss, N.~Almaadeed, and S.~Al-Maadeed, ``{A Review of Video Surveillance Systems},'' \emph{Journal of Visual Communication and Image Representation}, vol.~77, p. 103116, 2021.

\bibitem{kay2017kinetics}
W.~Kay, J.~Carreira, K.~Simonyan, B.~Zhang, C.~Hillier, S.~Vijayanarasimhan, F.~Viola, T.~Green, T.~Back, P.~Natsev \emph{et~al.}, ``{The Kinetics Human Action Video Dataset},'' \emph{CoRR abs/1705.06950}, 2017.

\bibitem{sigurdsson2016hollywood}
G.~A. Sigurdsson, G.~Varol, X.~Wang, A.~Farhadi, I.~Laptev, and A.~Gupta, ``{Hollywood in Homes: Crowdsourcing Data Collection for Activity Understanding},'' in \emph{ECCV}, 2016, pp. 510--526.

\bibitem{karpathy2014large}
A.~Karpathy, G.~Toderici, S.~Shetty, T.~Leung, R.~Sukthankar, and F.~Li, ``{Large-Scale Video Classification with Convolutional Neural Networks},'' in \emph{CVPR}, 2014, pp. 1725--1732.

\bibitem{2024youtube}
A.~Gilbertson and A.~Reisner, ``{Apple, Nvidia, Anthropic Used Thousands of Swiped YouTube Videos to Train AI},'' {https://www.proofnews.org/apple-nvidia-anthropic-used-thousands-of-swiped-youtube-videos-to-train-ai}, 2024.

\bibitem{sablayrolles2020radioactive}
A.~Sablayrolles, M.~Douze, C.~Schmid, and H.~J{\'e}gou, ``{Radioactive Data: Tracing through Training},'' in \emph{ICML}, 2020, pp. 8326--8335.

\bibitem{li2023black}
Y.~Li, M.~Zhu, X.~Yang, Y.~Jiang, T.~Wei, and S.-T. Xia, ``{Black-Box Dataset Ownership Verification via Backdoor Watermarking},'' \emph{IEEE TIFS}, vol.~18, pp. 2318--2332, 2023.

\bibitem{guo2023domain}
J.~Guo, Y.~Li, L.~Wang, S.-T. Xia, H.~Huang, C.~Liu, and B.~Li, ``{Domain Watermark: Effective and Harmless Dataset Copyright Protection is Closed at Hand},'' in \emph{NeurIPS}, 2023, pp. 54\,421--54\,450.

\bibitem{guo2025audio}
H.~Guo, J.~Guo, B.~Chen, Y.~Wang, X.~Chen, H.~Huang, Q.~Yan, and L.~Xiao, ``{AUDIO WATERMARK: Dynamic and Harmless Watermark for Black-Box Voice Dataset Copyright Protection},'' in \emph{USENIX Security}, 2025.

\bibitem{miao2021audio}
Y.~Miao, M.~Xue, C.~Chen, L.~Pan, J.~Zhang, B.~Z.~H. Zhao, D.~Kaafar, and Y.~Xiang, ``{The Audio Auditor: User-Level Membership Inference in Internet of Things Voice Services},'' in \emph{PoPETS}, 2021.

\bibitem{du2025sok}
L.~Du, X.~Zhou, M.~Chen, C.~Zhang, Z.~Su, P.~Cheng, J.~Chen, and Z.~Zhang, ``{SoK: Dataset Copyright Auditing in Machine Learning Systems},'' in \emph{IEEE S{\&}P}, 2025.

\bibitem{huang2024general}
Z.~Huang, N.~Z. Gong, and M.~K. Reiter, ``{A General Framework for Data-Use Auditing of ML Models},'' in \emph{ACM CCS}, 2024.

\bibitem{shokri2017membership}
R.~Shokri, M.~Stronati, C.~Song, and V.~Shmatikov, ``{Membership Inference Attacks Against Machine Learning Models},'' in \emph{IEEE S{\&}P}, 2017, pp. 3--18.

\bibitem{song2019auditing}
C.~Song and V.~Shmatikov, ``{Auditing Data Provenance in Text-Generation Models},'' in \emph{SIGKDD}, 2019, pp. 196--206.

\bibitem{chen2023face}
M.~Chen, Z.~Zhang, T.~Wang, M.~Backes, and Y.~Zhang, ``{FACE-AUDITOR: Data Auditing in Facial Recognition Systems},'' in \emph{USENIX Security}, 2023, pp. 7195--7212.

\bibitem{li2022untargeted}
Y.~Li, Y.~Bai, Y.~Jiang, Y.~Yang, S.-T. Xia, and B.~Li, ``{Untargeted Backdoor Watermark: Towards Harmless and Stealthy Dataset Copyright Protection},'' in \emph{NeurIPS}, 2022.

\bibitem{co2019procedural}
K.~T. Co, L.~Mu{\~n}oz-Gonz{\'a}lez, S.~de~Maupeou, and E.~C. Lupu, ``{Procedural Noise Adversarial Examples for Black-Box Attacks on Deep Convolutional Networks},'' in \emph{ACM CCS}, 2019, pp. 275--289.

\bibitem{carlini2022membership}
N.~Carlini, S.~Chien, M.~Nasr, S.~Song, A.~Terzis, and F.~Tramer, ``{Membership Inference Attacks from First Principles},'' in \emph{IEEE S{\&}P}, 2022, pp. 1897--1914.

\bibitem{zhang2022inference}
Z.~Zhang, M.~Chen, M.~Backes, Y.~Shen, and Y.~Zhang, ``{Inference Attacks Against Graph Neural Networks},'' in \emph{USENIX Security Symposium}, 2022, pp. 4543--4560.

\bibitem{liu2022ml}
Y.~Liu, R.~Wen, X.~He, A.~Salem, Z.~Zhang, M.~Backes, E.~De~Cristofaro, M.~Fritz, and Y.~Zhang, ``{ML-Doctor: Holistic Risk Assessment of Inference Attacks Against Machine Learning Models},'' in \emph{USENIX Security Symposium}, 2022, pp. 4525--4542.

\bibitem{maini2021dataset}
P.~Maini, M.~Yaghini, and N.~Papernot, ``{Dataset Inference: Ownership Resolution in Machine Learning},'' in \emph{ICLR}, 2021.

\bibitem{li2021membership}
Z.~Li and Y.~Zhang, ``{Membership Leakage in Label-Only Exposures},'' in \emph{ACM CCS}, 2021, pp. 880--895.

\bibitem{choquette2021label}
C.~A. Choquette-Choo, F.~Tramer, N.~Carlini, and N.~Papernot, ``{Label-Only Membership Inference Attacks},'' in \emph{ICML}, 2021.

\bibitem{tian2023knowledge}
Z.~Tian, Z.~Wang, A.~M. Abdelmoniem, G.~Liu, and C.~Wang, ``{Knowledge Representation of Training Data with Adversarial Examples Supporting Decision Boundary},'' \emph{IEEE TIFS}, 2023.

\bibitem{szyller2023robustness}
S.~Szyller, R.~Zhang, J.~Liu, and N.~Asokan, ``{On the Robustness of Dataset Inference},'' \emph{TMLR}, 2023.

\bibitem{du2024orl}
L.~Du, M.~Chen, M.~Sun, S.~Ji, P.~Cheng, J.~Chen, and Z.~Zhang, ``{ORL-AUDITOR: Dataset Auditing in Offline Deep Reinforcement Learning},'' in \emph{NDSS}, 2024.

\bibitem{liu2021encodermi}
H.~Liu, J.~Jia, W.~Qu, and N.~Z. Gong, ``{EncoderMI: Membership Inference Against Pre-trained Encoders in Contrastive Learning},'' in \emph{ACM CCS}, 2021, pp. 2081--2095.

\bibitem{song2021systematic}
L.~Song and P.~Mittal, ``{Systematic Evaluation of Privacy Risks of Machine Learning Models},'' in \emph{USENIX Security Symposium}, 2021, pp. 2615--2632.

\bibitem{du2025artistauditor}
L.~Du, Z.~Zhu, M.~Chen, Z.~Su, S.~Ji, P.~Cheng, J.~Chen, and Z.~Zhang, ``{ArtistAuditor: Auditing Artist Style Pirate in Text-to-Image Generation Models},'' in \emph{WWW}, 2025, pp. 2500--2513.

\bibitem{li2025vid}
Q.~Li, R.~Yu, and X.~Wang, ``{Vid-SME: Membership Inference Attacks Against Large Video Understanding Models},'' \emph{CoRR abs/2506.03179}, 2025.

\bibitem{sablayrolles2019white}
A.~Sablayrolles, M.~Douze, C.~Schmid, Y.~Ollivier, and H.~J{\'e}gou, ``{White-Box vs Black-Box: Bayes Optimal Strategies for Membership Inference},'' in \emph{ICML}, 2019, pp. 5558--5567.

\bibitem{dziedzic2022dataset}
A.~Dziedzic, H.~Duan, M.~A. Kaleem, N.~Dhawan, J.~Guan, Y.~Cattan, F.~Boenisch, and N.~Papernot, ``{Dataset Inference for Self-Supervised Models},'' in \emph{NeurIPS}, 2022.

\bibitem{li2022user}
G.~Li, S.~Rezaei, and X.~Liu, ``{User-Level Membership Inference Attack Against Metric Embedding Learning},'' in \emph{ICLR PAIR2Struct Workshop}, 2022.

\bibitem{dong2023rai2}
T.~Dong, S.~Li, G.~Chen, M.~Xue, H.~Zhu, and Z.~Liu, ``{RAI2: Responsible Identity Audit Governing the Artificial Intelligence},'' in \emph{NDSS}, 2023.

\bibitem{liu2022your}
G.~Liu, T.~Xu, X.~Ma, and C.~Wang, ``{Your Model Trains on My Data? Protecting Intellectual Property of Training Data via Membership Fingerprint Authentication},'' \emph{IEEE TIFS}, vol.~17, pp. 1024--1037, 2022.

\bibitem{salem2019ml}
A.~Salem, Y.~Zhang, M.~Humbert, P.~Berrang, M.~Fritz, and M.~Backes, ``{ML-Leaks: Model and Data Independent Membership Inference Attacks and Defenses on Machine Learning Models},'' in \emph{NDSS}, 2019.

\bibitem{wenger2024data}
E.~Wenger, X.~Li, B.~Y. Zhao, and V.~Shmatikov, ``{Data Isotopes for Data Provenance in DNNs},'' in \emph{PoPETS}, 2024, pp. 413--429.

\bibitem{guo2024zeromark}
J.~Guo, Y.~Li, R.~Chen, Y.~Wu, C.~Liu, and H.~Huang, ``{ZeroMark: Towards Dataset Ownership Verification without Disclosing Watermark},'' in \emph{NeurIPS}, 2024.

\bibitem{chen2025MembershipTracker}
Z.~Chen and K.~Pattabiraman, ``{Anonymity Unveiled: A Practical Framework for Auditing Data Use in Deep Learning Models},'' in \emph{ACM CCS}, 2025.

\bibitem{gu2019badnets}
T.~Gu, K.~Liu, B.~Dolan-Gavitt, and S.~Garg, ``{BadNets: Evaluating Backdooring Attacks on Deep Neural Networks},'' \emph{IEEE Access}, vol.~7, pp. 47\,230--47\,244, 2019.

\bibitem{li2021invisible}
Y.~Li, Y.~Li, B.~Wu, L.~Li, R.~He, and S.~Lyu, ``{Invisible Backdoor Attack with Sample-Specific Triggers},'' in \emph{CVPR}, 2021, pp. 16\,463--16\,472.

\bibitem{li2020open}
Y.~Li, Z.~Zhang, J.~Bai, B.~Wu, Y.~Jiang, and S.-T. Xia, ``{Open-Sourced Dataset Protection via Backdoor Watermarking},'' \emph{CoRR abs/2010.05821}, 2020.

\bibitem{li2022black}
Y.~Li, M.~Zhu, X.~Yang, Y.~Jiang, and S.-T. Xia, ``{Black-Box Ownership Verification for Dataset Protection via Backdoor Watermarking},'' \emph{CoRR abs/2209.06015}, 2022.

\bibitem{al2024look}
H.~A. Al~Kader~Hammoud, S.~Liu, M.~Alkhrashi, F.~Albalawi, and B.~Ghanem, ``{Look Listen and Attack: Backdoor Attacks Against Video Action Recognition},'' in \emph{CVPR Workshop}, 2024, pp. 3439--3450.

\bibitem{souri2022sleeper}
H.~Souri, L.~Fowl, R.~Chellappa, M.~Goldblum, and T.~Goldstein, ``{Sleeper Agent: Scalable Hidden Trigger Backdoors for Neural Networks Trained from Scratch},'' in \emph{NeurIPS}, 2022.

\bibitem{tang2023did}
R.~Tang, Q.~Feng, N.~Liu, F.~Yang, and X.~Hu, ``{Did You Train on My Dataset? Towards Public Dataset Protection with Cleanlabel Backdoor Watermarking},'' in \emph{ACM SIGKDD}, 2023.

\bibitem{bouaziz2025data}
W.~Bouaziz, N.~Usunier, and E.-M. El-Mhamdi, ``{Data Taggants: Dataset Ownership Verification via Harmless Targeted Data Poisoning},'' in \emph{ICLR}, 2025.

\bibitem{sadhu2021visual}
A.~Sadhu, T.~Gupta, M.~Yatskar, R.~Nevatia, and A.~Kembhavi, ``{Visual Semantic Role Labeling for Video Understanding},'' in \emph{CVPR}, 2021, pp. 5589--5600.

\bibitem{pareek2021survey}
P.~Pareek and A.~Thakkar, ``{A Survey on Video-Based Human Action Recognition: Recent Updates, Datasets, Challenges, and Applications},'' \emph{Artificial Intelligence Review}, vol.~54, no.~3, pp. 2259--2322, 2021.

\bibitem{parashar2023data}
A.~Parashar, A.~Parashar, W.~Ding, M.~Shabaz, and I.~Rida, ``{Data Preprocessing and Feature Selection Techniques in Gait Recognition: A Comparative Study of Machine Learning and Deep Learning Approaches},'' \emph{Pattern Recognition Letters}, vol. 172, pp. 65--73, 2023.

\bibitem{liang2023survey}
H.~Liang, Z.~Zhang, C.~Hu, Y.~Gong, and D.~Cheng, ``{A Survey on Spatio-Temporal Big Data Analytics Ecosystem: Resource Management, Processing Platform, and Applications},'' \emph{IEEE Transactions on Big Data}, vol.~10, no.~2, pp. 174--193, 2023.

\bibitem{lin2019tsm}
J.~Lin, C.~Gan, and S.~Han, ``{TSM: Temporal Shift Module for Efficient Video Understanding},'' in \emph{CVPR}, 2019, pp. 7083--7093.

\bibitem{luo2019grouped}
C.~Luo and A.~L. Yuille, ``{Grouped Spatial-Temporal Aggregation for Efficient Action Recognition},'' in \emph{CVPR}, 2019.

\bibitem{wang2016temporal}
L.~Wang, Y.~Xiong, Z.~Wang, Y.~Qiao, D.~Lin, X.~Tang, and L.~Van~Gool, ``{Temporal Segment Networks: Towards Good Practices for Deep Action Recognition},'' in \emph{ECCV}, 2016, pp. 20--36.

\bibitem{liu2016spatio}
J.~Liu, A.~Shahroudy, D.~Xu, and G.~Wang, ``{Spatio-Temporal LSTM with Trust Gates for 3d Human Action Recognition},'' in \emph{ECCV}, 2016.

\bibitem{feichtenhofer2020x3d}
C.~Feichtenhofer, ``{X3D: Expanding Architectures for Efficient Video Recognition},'' in \emph{CVPR}, 2020, pp. 203--213.

\bibitem{feichtenhofer2019slowfast}
C.~Feichtenhofer, H.~Fan, J.~Malik, and K.~He, ``{SlowFast Networks for Video Recognition},'' in \emph{CVPR}, 2019, pp. 6202--6211.

\bibitem{tran2018closer}
D.~Tran, H.~Wang, L.~Torresani, J.~Ray, Y.~LeCun, and M.~Paluri, ``{A Closer Look at Spatiotemporal Convolutions for Action Recognition},'' in \emph{CVPR}, 2018, pp. 6450--6459.

\bibitem{tran2019video}
D.~Tran, H.~Wang, L.~Torresani, and M.~Feiszli, ``{Video Classification with Channel-Separated Convolutional Networks},'' in \emph{CVPR}, 2019, pp. 5552--5561.

\bibitem{carreira2017quo}
J.~Carreira and A.~Zisserman, ``{Quo Vadis, Action Recognition? A New Model and The Kinetics Dataset},'' in \emph{CVPR}, 2017.

\bibitem{liu2022video}
Z.~Liu, J.~Ning, Y.~Cao, Y.~Wei, Z.~Zhang, S.~Lin, and H.~Hu, ``{Video Swin Transformer},'' in \emph{CVPR}, 2022, pp. 3202--3211.

\bibitem{bertasius2021space}
G.~Bertasius, H.~Wang, and L.~Torresani, ``{Is Space-Time Attention All You Need for Video Understanding?}'' in \emph{ICML}, 2021.

\bibitem{arnab2021vivit}
A.~Arnab, M.~Dehghani, G.~Heigold, C.~Sun, M.~Lu{\v{c}}i{\'c}, and C.~Schmid, ``{ViViT: A Video Vision Transformer},'' in \emph{CVPR}, 2021, pp. 6836--6846.

\bibitem{lagae2010survey}
A.~Lagae, S.~Lefebvre, R.~Cook, T.~DeRose, G.~Drettakis, D.~S. Ebert, J.~P. Lewis, K.~Perlin, and M.~Zwicker, ``{A Survey of Procedural Noise Functions},'' in \emph{Computer Graphics Forum}, vol.~29, no.~8, 2010, pp. 2579--2600.

\bibitem{perlin2002improving}
K.~Perlin, ``{Improving Noise},'' in \emph{ACM SIGGRAPH}, 2002.

\bibitem{woolson2005wilcoxon}
R.~F. Woolson, ``{Wilcoxon Signed-Rank Test},'' \emph{Encyclopedia of Biostatistics}, vol.~8, 2005.

\bibitem{kuehne2011hmdb}
H.~Kuehne, H.~Jhuang, E.~Garrote, T.~Poggio, and T.~Serre, ``{HMDB: A Large Video Database for Human Motion Recognition},'' in \emph{ICCV}, 2011, pp. 2556--2563.

\bibitem{soomro2012ucf101}
K.~Soomro, A.~R. Zamir, and M.~Shah, ``{UCF101: A Dataset of 101 Human Actions Classes from Videos in the Wild},'' \emph{CoRR/abs:1212.0402}, 2012.

\bibitem{goyal2017something}
R.~Goyal \emph{et~al.}, ``{The `Something Something' Video Database for Learning and Evaluating Visual Common Sense},'' in \emph{ICCV}, 2017, pp. 5842--5850.

\bibitem{chen2021deep}
C.-F.~R. Chen, R.~Panda, K.~Ramakrishnan, R.~Feris, J.~Cohn, A.~Oliva, and Q.~Fan, ``{Deep Analysis of Cnn-Based Spatio-Temporal Representations for Action Recognition},'' in \emph{CVPR}, 2021.

\bibitem{2020mmaction2}
M.~Contributors, ``{OpenMMLab's Next Generation Video Understanding Toolbox and Benchmark},'' {https://github.com/open-mmlab/mmaction2}, 2020.

\bibitem{Jordan2024on}
K.~Jordan, ``{On the Variance of Neural Network Training with respect to Test Sets and Distributions},'' in \emph{ICLR}, 2024.

\bibitem{hendrycks2016baseline}
D.~Hendrycks and K.~Gimpel, ``{A Baseline for Detecting Misclassified and Out-of-Distribution Examples in Neural Networks},'' in \emph{ICLR}, 2017.

\bibitem{sun2022out}
Y.~Sun, Y.~Ming, X.~Zhu, and Y.~Li, ``{Out-of-Distribution Detection with Deep Nearest Neighbors},'' in \emph{ICML}, 2022.

\bibitem{hendrycks2022scaling}
D.~Hendrycks, S.~Basart, M.~Mazeika, A.~Zou, J.~Kwon, M.~Mostajabi, J.~Steinhardt, and D.~Song, ``{Scaling Out-of-Distribution Detection for Real-World Settings},'' in \emph{ICML}, 2022.

\bibitem{lee2018simple}
K.~Lee, K.~Lee, H.~Lee, and J.~Shin, ``{A Simple Unified Framework for Detecting Out-of-Distribution Samples and Adversarial Attacks},'' in \emph{NeurIPS}, 2018.

\bibitem{wang2022vim}
H.~Wang, Z.~Li, L.~Feng, and W.~Zhang, ``{ViM: Out-of-Distribution with Virtual-Logit Matching},'' in \emph{CVPR}, 2022, pp. 4921--4930.

\end{thebibliography}
    \bibliographystyle{IEEEtran}
}

\appendix

\subsection{Threshold Analysis}
\label{subsec:appendix_threshold_analysis}

\begin{theorem}[Threshold range under TPR and FPR Constraints]
    Let $\bar h$ denote the average probability difference between two models over $n$ samples. 
    Assume that under the null hypothesis (\ie, same training dataset), $\bar h \sim \mathcal{N}(\mu_0,\sigma_0^2/n)$;
    and under the alternative hypothesis (\ie, different training datasets),
    $\bar h \sim \mathcal{N}(\mu_1,\sigma_1^2/n)$, with $\mu_0 < \mu_1$.
    For given significance levels $a$ (FPR) and $b$ (FNR), any threshold $\tau$ satisfying 
    \begin{equation*}
        \tau \in [\mu_0+z_{1-b}\frac{\sigma_0}{\sqrt {n}},\mu_1+z_a\frac{\sigma_1}{\sqrt{n}}]
    \end{equation*}
    ensuring that $TPR=P(\bar h \leq \tau | H_0)\geq 1-b$ and $FPR=P(\bar h \leq \tau | H_1)\leq a$.
    \label{theorem_threshold}
\end{theorem}

\begin{proof}
    The TPR and FPR can be expressed as follows:
    \begin{align*}
        TPR(\tau)=P(\bar h \leq \tau | H_0)&=\Phi(\frac{\tau-\mu_0}{\sigma_0/\sqrt{n}}), \\
        FPR(\tau)=P(\bar h \leq \tau | H_1)&=\Phi(\frac{\tau-\mu_1}{\sigma_1/\sqrt{n}}),
    \end{align*}
where $\Phi(\cdot)$ is the cumulative distribution function of the standard normal distribution.
According to the constraints of $TPR\geq 1-b$
and $FPR \leq a$,
we can obtain
\begin{align*}
    \frac{\tau-\mu_0}{\sigma_0/\sqrt{n}} &\geq z_{1-b}, \\
    \frac{\tau-\mu_1}{\sigma_1/\sqrt{n}} &\leq z_a,
\end{align*}
which directly yield the bounds as follows:
\begin{equation*}
    \tau_{min} = \mu_0 + z_{1-b}\frac{\sigma_0}{\sqrt{n}},
    \tau_{max} = \mu_1 + z_{a}\frac{\sigma_1}{\sqrt{n}}.
\end{equation*}

\end{proof}

\mypara{Determination of the upper threshold}
In practice, due to the diversity of video models and datasets, it is difficult to obtain accurate values of the above parameters.
According to existing research~\cite{feichtenhofer2019slowfast,carreira2017quo,Jordan2024on},
the accuracy of training with the same dataset fluctuates by about 1\%-2\%, while the migration between different training datasets can cause the accuracy deviation to exceed 5\% or even 10\%.
Therefore, here we assume that $\mu_0=0.02,\sigma_0=0.01,\mu_1=0.08,\sigma_1=0.02,n=100$, let $a=b=0.05$, we can calculate $\tau_{min}\approx0.022$ and $\tau_{max}\approx0.077$.
To achieve a suitable trade-off between TPR and FPR, we choose the upper threshold limit $H$ as the midpoint between the minimum and maximum values, \ie, 
$H=(\tau_{min}+\tau_{max})/2\approx0.05$.

\subsection{Theoretical Bound of FPR}
\label{subsec:FPR_theoretical_bound}

\mypara{Notations and Assumptions}
Let $\{\Delta s_{R,i}\}_{i=1}^{n_R}$ denote the pairwise score differences on the reference set used to estimate the threshold:
\[
\bar h = \frac{1}{n_R}\sum_{i=1}^{n_R}\Delta s_{R,i},\qquad
h = {clip}(\bar h,-H,H).
\]
Let $\{\Delta s_{M,i}\}_{i=1}^{n_M}$ denote the non-zero paired differences on the modification set,
for $i\in\{1,\cdots,n_M\}$,
we have
\[d_i = h - \Delta s_{M,i},\]
and let $W$ be the Wilcoxon signed-rank statistic (the sum of ranks of positive $d_i$’s).
In the absence of ties,
\[
\mu_W = \frac{n_M(n_M+1)}{4},\quad
\sigma_W^2 = \frac{n_M(n_M+1)(2n_M+1)}{24}.
\]

To obtain a bounded conclusion, we impose the following suitable assumptions.

\begin{enumerate}
    \item The reference mean $\bar h$ is sub-Gaussian around its population mean $\mu$,
    \ie, for any $\delta_h>0$,
    \[
    \Pr(|\bar h-\mu|>\delta_h)\le 2\exp\!\Big(-\frac{n_R\delta_h^2}{2c_h^2}\Big),
  \]
  where $c_h$
  is a theoretical
  constant characterizing the concentration of the estimated reference threshold
  $\bar h$.
  A smaller $c_h$ implies tighter concentration and a smaller FPR correction term.
  The value of $c_h$ can be estimated as $s_R / \sqrt{n_R}$, where $s_R^2$ is their sample variance.
  This assumption is suitable since each
    \(\Delta s_{R,i}\) is a bounded or light-tailed random variable representing a difference
    between two model outputs on the reference set.
    According to Hoeffding Lemma,
    if \(|\Delta s_{R,i}|\le M\),
    then \(\Delta s_{R}\) is \(M\)-sub-Gaussian.
\label{assump:sub-Gaussian}
    
    \item 

    The density \(f_{\Delta}\) of the modification-set differences \(\Delta s_{M,i}\) is uniformly bounded near \(\mu\), \ie, there exists \(f_{\max}>0\) such that \(f_{\Delta}(x)\le f_{\max}\) for \(x\) in a neighborhood of \(\mu\).

    \item The Algorithm 3 modifies at most $k_{\mathrm{pp}}$ of the $n_M$ samples through post-processing.
    \label{assump:k_pp}

    \item 
    According to the asymptotic normality,
the Wilcoxon signed-rank statistic $W$ is asymptotically normally
distributed under the null hypothesis $H_0$ when the number of matched samples $n_M$ is large.
\label{assump:asym_norm}
 
\end{enumerate}

\mypara{Auxiliary Lemmas}
To obtain the final conclusion, we further introduce several auxiliary lemmas as follows.

\begin{lemma}[Event decomposition]
\label{lem:decomp_clip}
The overall FPR can be decomposed into three %
cases.
For any $\delta_h>0$,
\[
\begin{aligned}
\Pr_{H_0}(\text{reject}) &\le
\Pr(\text{reject}\mid |\bar h-\mu|\le\delta_h,\,|\bar h|\le H) \\
&+\Pr(|\bar h-\mu|>\delta_h)
+\Pr(|\bar h|>H).
\end{aligned}
\]
\end{lemma}

\begin{lemma}[Affected-sample bounds]
\label{lem:affected_clip}
If $|\bar h-\mu|\le\delta_h$ and $|\bar h|\le H$, then 
the maximum number of samples $k(\delta_h)$ that can change sign or rank is  
\[
k(\delta_h)= n_M \int_{\mu-|\delta_h|}^{\mu+|\delta_h|} f_\Delta(x)\,dx \le n_M \min\{1,\,2f_{\max}\delta_h\}.
\]

Similarly, if clipping occurs (\ie, $|\bar h|>H$), then the maximum affected samples $k_{\mathrm{clip}}$ is bounded by
\[
k_{\mathrm{clip}} \le n_M \min\{1,\,2f_{\max}(H-|\mu|)\}.
\]
\end{lemma}

\begin{lemma}[Rank-sum perturbation bound]
\label{lem:rankperturb_clip}
If at most $K$ samples are arbitrarily changed or permuted, the Wilcoxon statistic $W$
can change by at most
\[
\Delta W_{\max}(K) = \frac{K(2n_M - K + 1)}{2}.
\]
\end{lemma}

\mypara{Main Theorem}
We can obtain the following theorem:

\begin{theorem}[FPR bound]
\label{thm:main_clip}
Based on the above assumptions and lemmas, for any $\delta_h>0$, the false positive rate of the Wilcoxon-based test
satisfies
\begin{equation}\label{eq:main_bound}
\begin{aligned}
\mathrm{FPR}\;\le\;&
\alpha 
+ \frac{\Delta W_{\max}\big(k(\delta_h)+k_{\mathrm{pp}}+k_{\mathrm{clip}}\big)}{\sigma_W\sqrt{2\pi}}\\
&\;+\; 
2\exp\!\Big(-\frac{n_R\delta_h^2}{2c_h^2}\Big)
+ 2\exp\!\Big(-\frac{n_R(H-|\mu|)^2}{2c_h^2}\Big).
\end{aligned}
\end{equation}

\end{theorem}

\begin{proof}

Applying Lemma~\ref{lem:decomp_clip},
\[
\begin{aligned}
\Pr(\text{reject})
&\le \Pr(\text{reject}\mid |\bar h-\mu|\le\delta_h,\,|\bar h|\le H) \\
&+ \Pr(|\bar h-\mu|>\delta_h)
+ \Pr(|\bar h|>H).
\end{aligned}
\]

On the event $\{|\bar h-\mu_0|\le\delta_h,\,|\bar h|\le H\}$,
at most $k(\delta_h)+k_{\mathrm{pp}}$ samples are affected
according to Lemma~\ref{lem:affected_clip} and
Assumption~\ref{assump:k_pp}.
By Assumption~\ref{assump:asym_norm}, the standardized
Wilcoxon statistic is asymptotically normal under $H_0$,
thus its rejection probability can be approximated as
\[
1 - \Phi\!\Big(z_\alpha - 
\frac{\Delta W_{\max}(k(\delta_h)+k_{\mathrm{pp}})}{\sigma_W}\Big),
\]
where $\Phi(\cdot)$ denotes the cumulative distribution function of
the standard normal distribution, and $z_\alpha$ is the $(1-\alpha)$
quantile satisfying $\Pr(Z>z_\alpha)=\alpha$ for
$Z\sim\mathcal{N}(0,1)$.
Applying the inequality $\Phi(a+b)\le\Phi(a)+b/\sqrt{2\pi}$
with $a=z_\alpha$ and rearranging terms gives
\[
\begin{aligned}
\Pr(\text{reject}\mid |\bar h-\mu_0|\le\delta_h,\,|\bar h|\le H)
&\le 
\alpha \\
+ 
\frac{\Delta W_{\max}\big(k(\delta_h)+k_{\mathrm{pp}}\big)}
{\sigma_W\sqrt{2\pi}}.
\end{aligned}
\]

If clipping occurs (\ie, $|\bar h|>H$), additional at most $k_{\mathrm{clip}}$ samples may be
affected.  
Including this in the rank perturbation budget, we further obtain
\[
\alpha
+ \frac{\Delta W_{\max}(k(\delta_h)+k_{\mathrm{pp}}+k_{\mathrm{clip}})}{\sigma_W\sqrt{2\pi}}.
\]

According to Assumption \ref{assump:sub-Gaussian},
the last two probability terms of \autoref{eq:main_bound} can be bounded by
\begin{equation*}
\begin{aligned}
     \Pr(|\bar h-\mu|>\delta)&\le 2\exp\!\Big(-\frac{n_R\delta_h^2}{2c_h^2}\Big), \\
    \Pr(|\bar h|>H) & \le 2\exp\!\Big(-\frac{n_R(H-|\mu|)^2}{2c_h^2}\Big).
\end{aligned}
\end{equation*}

Based on the above derivation, we can obtain the final FPR bound.
 Overall, the bound explicitly separates contributions from:
    (i) the nominal test level $\alpha$,
    (ii) rank perturbation due to threshold estimation, post-processing, and clipping,
    and (iii) the probabilities of large threshold deviation and clipping.

In practical deployment, the significance level $\alpha$ and the clipping threshold $H$ are manually specified parameters. 
The standard deviation $\sigma_W$ and the sample sizes $n_R$ and $n_M$ can be directly computed from the data, whereas 
$\delta_h$, $c_h$, and $\mu_h$ are theoretical parameters determined by the underlying assumptions. 
The total number of potentially affected samples is upper bounded by $(k(\delta_h)+k_{\mathrm{pp}}+k_{\mathrm{clip}})$. 
Given these quantities, the theoretical upper bound of the FPR can be computed accordingly. 
Note that the final calculated result is the upper bound of FPR, and the actual FPR may be much smaller than it.

\end{proof}

\subsection{Discussion of Other Methods}
\label{subsec:appendix_discussion}

In this work, we apply the SOTA auditing methods (\ie, \mlda~\cite{huang2024general} and \mt~\cite{chen2025MembershipTracker}) for image data to the video domain as the baselines.
For these methods, we process each
frame in the video separately and then merge them together.

The idea of \mlda is to generate two images as far apart as
possible for the same sample and then randomly select one to
publish, and finally use the suspect model’s output on the two
generated images for auditing.
On the one hand, 
this method requires a large modification ratio in the published version.
On the other hand,
video models are more robust in terms of recognition accuracy,
making this method less effective.

For~\mt, a data marking technique based on image blending and noise injection is designed to mark target data.
Then, a membership inference-based inference process is utilized to achieve the verification by observing the loss values.
Though \mt only requires to mark a small fraction of data,
it obviously degrades the marked image quality,
which reduces the stealthiness.
In addition, due to the complexity of the video model and the diversity of video data, the outputs of member samples and some non-member samples are prone to present similar outputs (such as both high or low), which can interfere with the audit results.
Moreover, for \mt, calculating the threshold based on the non-member loss distribution cannot consistently yield promising results across different video models and datasets. 
Furthermore, video data consists of a sequence of frames, and injecting different noise into consecutive frames independently can easily disrupt the coherence of the video.

Many gradient optimization-based auditing methods~\cite{sablayrolles2020radioactive,li2022untargeted} designed for image data are difficult to directly apply to video.
This is because raw videos typically undergo frame sampling and cropping, which disrupt the input continuity. 
As a result, generating effective perturbations for the entire video based on limited input information becomes highly challenging.
In addition, many backdoor-based methods require modifying the true labels of the original data, which is often impractical in real-world scenarios. 
Such approaches are prone to introducing harmful side effects and potential security risks.

\subsection{Robustness of Input Detection}
\label{subsec:appendix_robustness}

A malicious attacker may deploy outlier detection methods to filter the training data.
Here,
we apply multiple common detection approaches including MSP~\cite{hendrycks2016baseline}, KNN~\cite{sun2022out}, KLM~\cite{hendrycks2022scaling}, Mahalanobis~\cite{lee2018simple} and ViM~\cite{wang2022vim}.
From different perspectives such as feature space, output probability and logits, the above methods achieve outlier detection based on the statistical characteristics of the normal sample distribution.
A pre-trained R3D-18 model is utilized to obtain the feature representation, and the malicious  attacker has access to a small set of unmodified samples to fit the detector.

\autoref{table:OOD_detection} illustrates the auditing performance for the I3D model on the UCF-101 dataset when the attacker applies the input detection mechanism.
We find that \method can still achieve the ideal auditing performance under input detection since these detection methods cannot effectively remove modified samples from the training set.
This highlights that \method is robust to common outlier detection methods.

\begin{table}[htb]
    \caption{Auditing performance under input detection.}
    \vspace{-0.1cm}
    \centering
    \begin{tabular}{c | c | c | c | c | c }
    \toprule
     \textbf{Method} & \textbf{MSP} & \textbf{KNN} &  \textbf{KLM} & \textbf{Mahalanobi} & \textbf{ViM}
     \\
     \midrule
       TPR & $1.000$ & $1.000$ & $1.000$ & $1.000$ & $1.000$ \\
       FPR & $0.000$ & $0.000$ & $0.000$ & $0.000$ & $0.000$ \\
       
      \bottomrule
    \end{tabular}
    \label{table:OOD_detection}
\end{table}

\subsection{Parameter Variation}
\label{subsec:appendix_parameter_variation}

\subsubsection{Impact of Perturbation Budget}
\label{subsubsec:appendix_impact_pert_budget}

\autoref{fig:param_vary_epsilon_hmdb51} illustrates the impact of various perturbation budgets for three suspect models on the HMDB-51 dataset.
We find that even with a perturbation budget of only 2, \method could still achieve a TPR of 1 on the I3D and SlowFast models while maintaining a low FPR. 
This is because the HMDB51 dataset is relatively small in scale, and the target model’s performance is generally lower than on the UCF-101 dataset. 
In such a setting, \method is more likely to succeed in amplifying the effect of the modified samples. 
The TSM model exhibits a lower TPR compared to the other two models. 
We speculate that this is due to fewer training epochs, making the impact of a low perturbation budget less noticeable. 
As the value of $\varepsilon$ increases, the auditing performance of \method continues to improve. 
When $\varepsilon$ reaches 4, \method achieves 100\% auditing accuracy on both the I3D and SlowFast models, further demonstrating its effectiveness. Additionally, the modified samples lead to the most significant performance drop on the SlowFast model, likely because it is trained from scratch without a pre-trained backbone, resulting in higher instability during training.
We further provide the visualization of generated videos under different perturbation budgets for HMDB-51, as shown in~\autoref{fig:vary_epsilon_video_show2}.

\begin{figure*}[!t]
    \centering
    \includegraphics[width=0.85\textwidth]{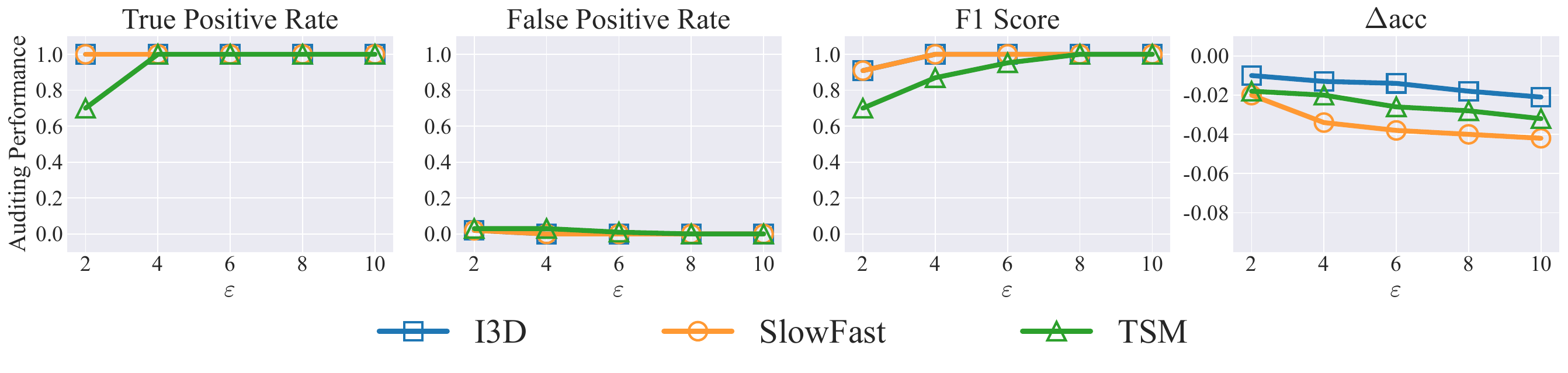}
    \vspace{-0.3cm}
    \caption{
    The impact of perturbation budget for three various suspect models on the four metrics of the HMDB-51 dataset.
    }
    \label{fig:param_vary_epsilon_hmdb51}
\vspace{-0.35cm}
\end{figure*}

\subsubsection{Impact of Modification Ratio}
\label{subsubsec:appendix_impact_modify_ratio}
\autoref{fig:param_vary_ratio_hmdb51} presents the auditing performance under different modification ratios on the HMDB-51 dataset. 
We observe that when the modification ratio is 0.5\%, the TPR of the three models does not reach 1, remaining around 0.6–0.8. 
This is primarily because HMDB-51 is relatively small in size, and a 0.5\% modification ratio corresponds to only about ten samples. 
Under such conditions, achieving high-precision auditing is challenging, as numerous factors influence model training, and the amplification effect of only a few modified samples tends to be unstable. 
Nevertheless, \method still achieves a TPR above 0.5 and an FPR below 0.05 across all three suspect models, demonstrating its strong auditing capability. 
As the modification ratio increases to 1\% and 2\%, \method can achieve a 100\% detection accuracy (\ie, TPR = 1 and FPR = 0), further highlighting its effectiveness. 

\begin{figure*}[!t]
    \centering
    \includegraphics[width=0.85\textwidth]{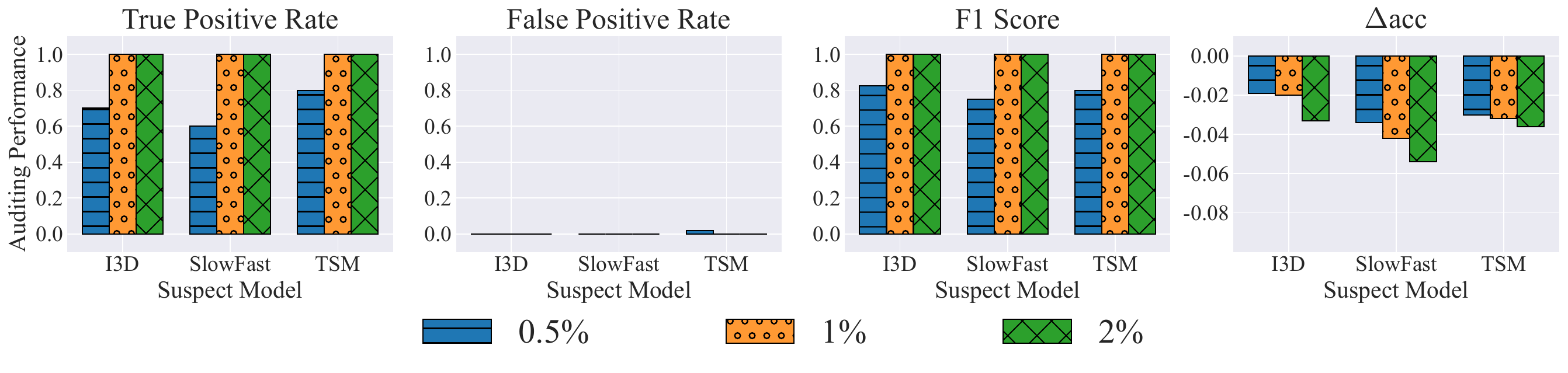}
    \vspace{-0.3cm}
    \caption{
    The impact of modification ratio for three 
    various 
    suspect models on 
    the 
    four metrics of the HMDB-51 dataset.
    }
    \label{fig:param_vary_ratio_hmdb51}
\vspace{-0.2cm}
\end{figure*}

\subsubsection{Impact of Various Evaluation Models}
\label{subsubsec:appendix_impact_various_eval_model}

\autoref{fig:param_vary_eval_model_ucf101} presents the results for four metrics on the UCF101 dataset using different evaluation models. 
We observe that when I3D and TSM are used as evaluation models, both the TPR and FPR achieve optimal performance. 
In contrast, when SlowFast is applied, the auditing accuracy slightly declines. 
This is likely because 
SlowFast
is trained from scratch without pre-trained weights, resulting in weaker generalization ability. 
Consequently, the selected samples are less effective compared to those identified by the other two models. 
These findings suggest that stronger evaluation models lead to more appropriate sample selection, thereby enhancing the final auditing.

\begin{figure*}[!t]
    \centering
    \includegraphics[width=0.85\textwidth]{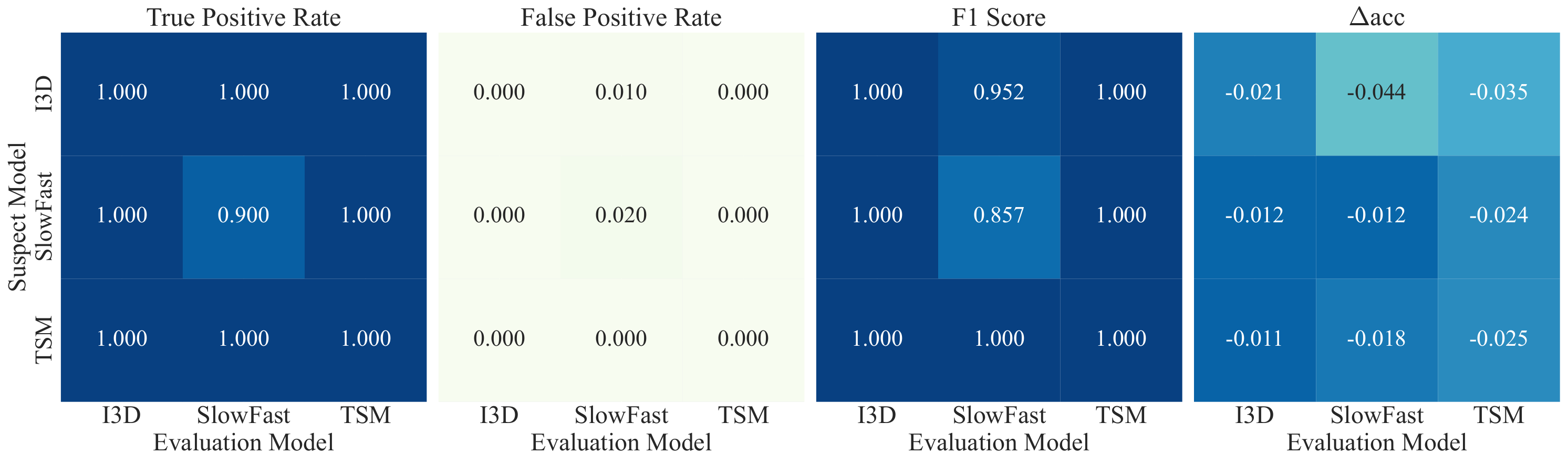}
    \vspace{-0.3cm}
    \caption{
    The impact of different evaluation models for three %
    suspect models on four metrics of the UCF-101 dataset.
    }
    \label{fig:param_vary_eval_model_ucf101}
\vspace{-0.2cm}
\end{figure*}

\autoref{fig:param_vary_eval_model_hmdb51}
presents the auditing results of different evaluation models for various suspect models on the HMDB-51 dataset. 
First, we observe that all evaluation models achieve 100\% detection accuracy across different suspect models, fully demonstrating the superiority and robustness of~\method. 
Second, we find that when the suspect model is SlowFast, the performance on normal tasks degrades most significantly. 
This is mainly because, unlike the other two models, SlowFast is trained from scratch without leveraging a pre-trained model, making its training more unstable.

\begin{figure*}[!t]
    \centering
    \includegraphics[width=0.85\textwidth]{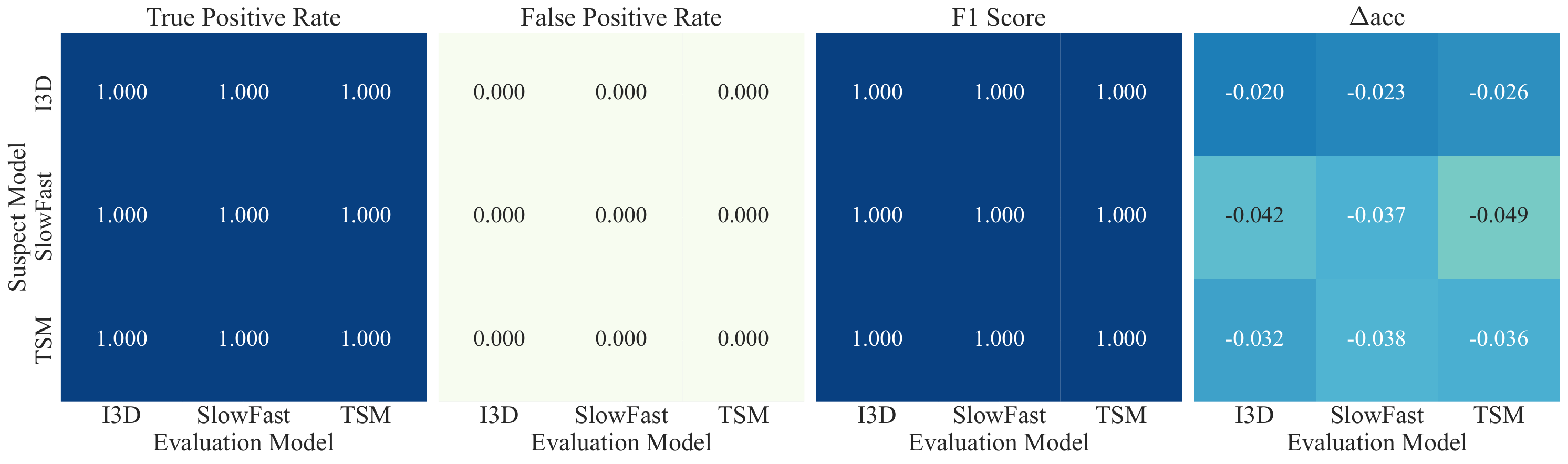}
    \vspace{-0.3cm}
    \caption{
    The impact of different evaluation models for three 
    suspect models on 
    four metrics of the HMDB-51 dataset.
    }
    \label{fig:param_vary_eval_model_hmdb51}
\vspace{-0.4cm}
\end{figure*}

\subsubsection{Impact of Noise Setting}
\label{subsubsec:appendix_impact_noise_setting}

In this section, we investigate the impact of different noise parameter settings on auditing performance. \autoref{table:varrious_noise_param} lists four parameter configurations (denoted as S1, S2, S3, and S4).
The parameters $\lambda_x$ and $\lambda_y$ control the spatial variation scale, while $\lambda_t$ determines the variation along the temporal axis. 
The parameter $\phi_{sine}$ introduces nonlinear mapping of the noise, and $\Omega$ is utilized to overlay multiple layers of noise with varying frequencies. 

\autoref{fig:param_vary_noise_ucf101} presents the TPR and FPR results on the UCF101 dataset under each configuration. 
We observe that \method achieves ideal auditing performance across all settings, consistently reaching TPR = 1 and FPR = 0. 
These results demonstrate the robustness and practicality of \method. 
In real-world deployment, dataset owners can flexibly adjust noise parameters according to specific requirements.

\begin{table}[!t]
    \centering
    \caption{Summary of various noise parameter settings.}
    \label{table:varrious_noise_param}
    \vspace{-0.1cm}
    \footnotesize
    \setlength{\tabcolsep}{1.2em}
	\begin{tabular}{cc}
		\toprule
		\textbf{Type} & \textbf{Setting}  \\
		\midrule
            
		S1 & $\lambda_{x}=\lambda_y=32,\lambda_t=6.4,\phi_{sine}=1,\Omega=2$ \\
		S2 & $\lambda_{x}=\lambda_y=32,\lambda_t=3.2,\phi_{sine}=1,\Omega=4$  \\
             S3 & $\lambda_{x}=\lambda_y=16,\lambda_t=3.2,\phi_{sine}=0.7,\Omega=5$  \\
            S4 & $\lambda_{x}=\lambda_y=16,\lambda_t=1.6,\phi_{sine}=1,\Omega=2$ \\
        
		\bottomrule
	\end{tabular}
    \vspace{-0.2cm}
\end{table}

\begin{figure}[!t]
    \centering
    \includegraphics[width=0.35\textwidth]{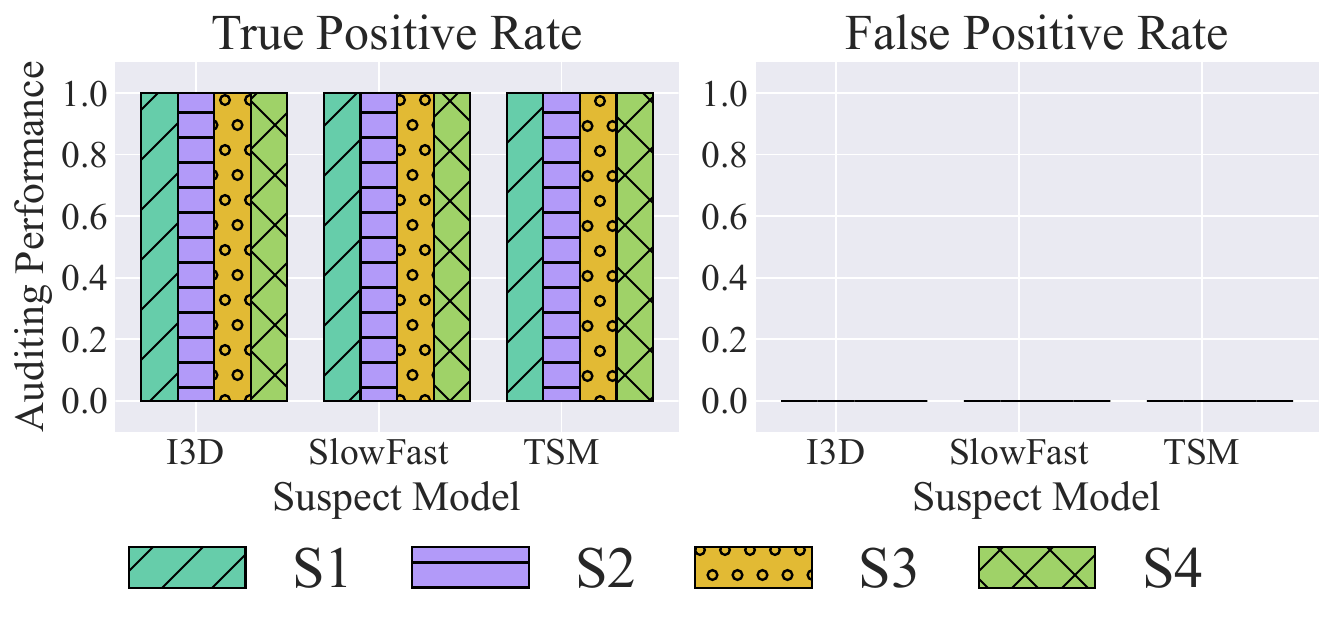}
    \vspace{-0.2cm}
    \caption{
    The impact of different noise settings for three 
    suspect models on two metrics of the UCF-101 dataset.
    }
    \label{fig:param_vary_noise_ucf101}
\vspace{-0.2cm}
\end{figure}

\begin{figure}[!htbp]
    \centering
    \includegraphics[width=0.35\textwidth]{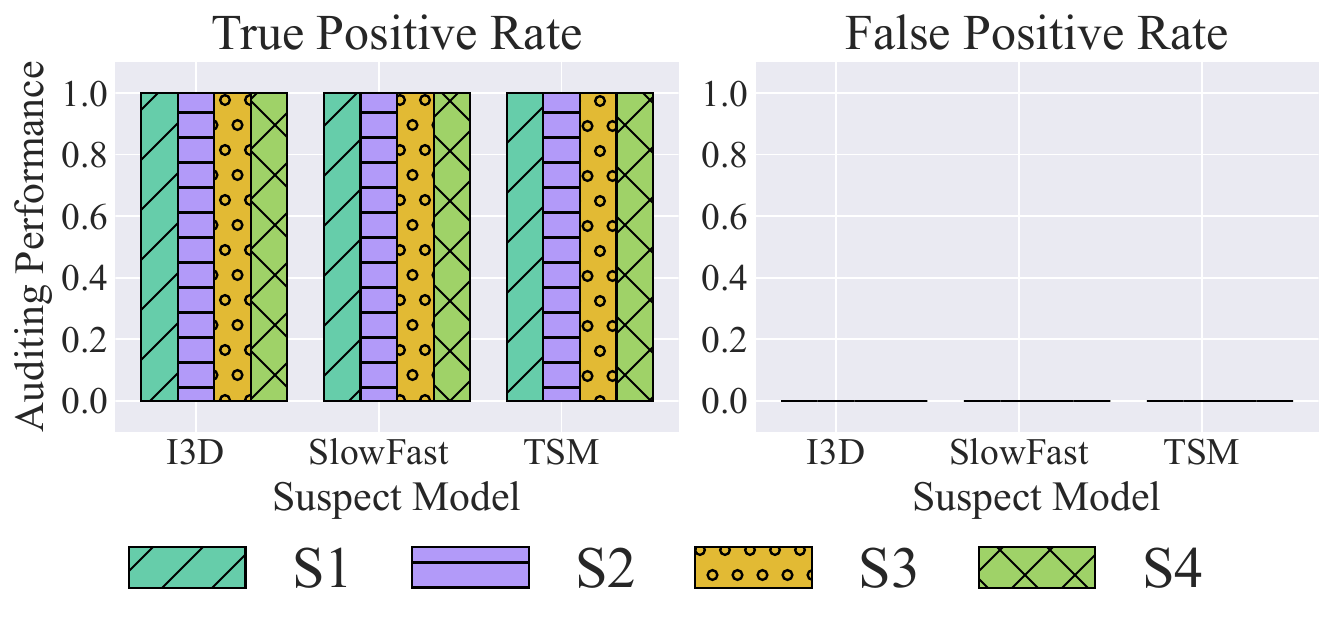}
    \vspace{-0.2cm}
    \caption{
    The impact of different noise settings for three %
    suspect models on 
    two metrics of the HMDB-51 dataset.
    }
    \label{fig:param_vary_noise_hmdb51}
\vspace{-0.2cm}
\end{figure}

\autoref{fig:param_vary_noise_hmdb51}
shows the TPRs and FPRs of three suspect models on the HMDB-51 dataset under different noise parameter settings. 
It can be seen that \method consistently achieves 100\% auditing accuracy (\ie, TPR = 1 and FPR = 0) across all suspect models. 
By injecting three-dimensional Perlin noise, the influence of the modified samples is effectively amplified, ensuring \method to accurately determine whether the published dataset has been misused based on the model's behavior. 
The above results and analysis demonstrate the versatility and robustness of~\method. 
In practical scenarios, the dataset owner can flexibly set the noise parameters according to their specific requirements.
We also provide an illustration of generated videos under various noise parameter settings for the HMDB-51 dataset in~\autoref{fig:vary_perlin_video_show2}.

\subsubsection{Impact of Threshold Setting}
\label{subsubsec:appendix_impact_threshold_setting}

As discussed in~\autoref{subsec:copyright_verify}, we apply a threshold clipping mechanism during hypothesis testing. 
In this section, we examine how different settings of the threshold upper bound $H$ affect the audit results. 
Since \method performs well under various thresholds when the perturbation budget $\varepsilon$ is 10, here we focus on a more challenging scenario and present the changes in TPR and FPR for three suspect models on the UCF-101 dataset when the perturbation budget $\varepsilon$ is set to 4.

As shown in~\autoref{fig:param_vary_clip_th_ucf101},
when $H$ is set to a very small value (\eg, 0), it imposes a strict condition where the output probability of the original sample must be significantly lower than that of the modified sample for a dataset misuse to be detected. 
In this case, the FPR of all three suspect models is 0, but the TPR drops to 0.6. 
Conversely, when $H$ is set to a larger value (\eg, 0.2), the TPR increases, but the FPR also rises significantly. These results indicate that setting
$H$ too low or too high leads to suboptimal auditing outcomes, \ie, either a reduced detection rate (low TPR) or increased false alarms (high FPR).
Therefore, selecting an appropriate value for $H$ is critical to balancing TPR and FPR. 
We find that 
$H=0.05$ provides a good trade-off between these two aspects, which is consistent with our analysis in~\autoref{subsec:appendix_threshold_analysis}.
Therefore, we adopt this value as the threshold upper bound. %

\begin{figure}[!t]
    \centering
    \includegraphics[width=0.35\textwidth]{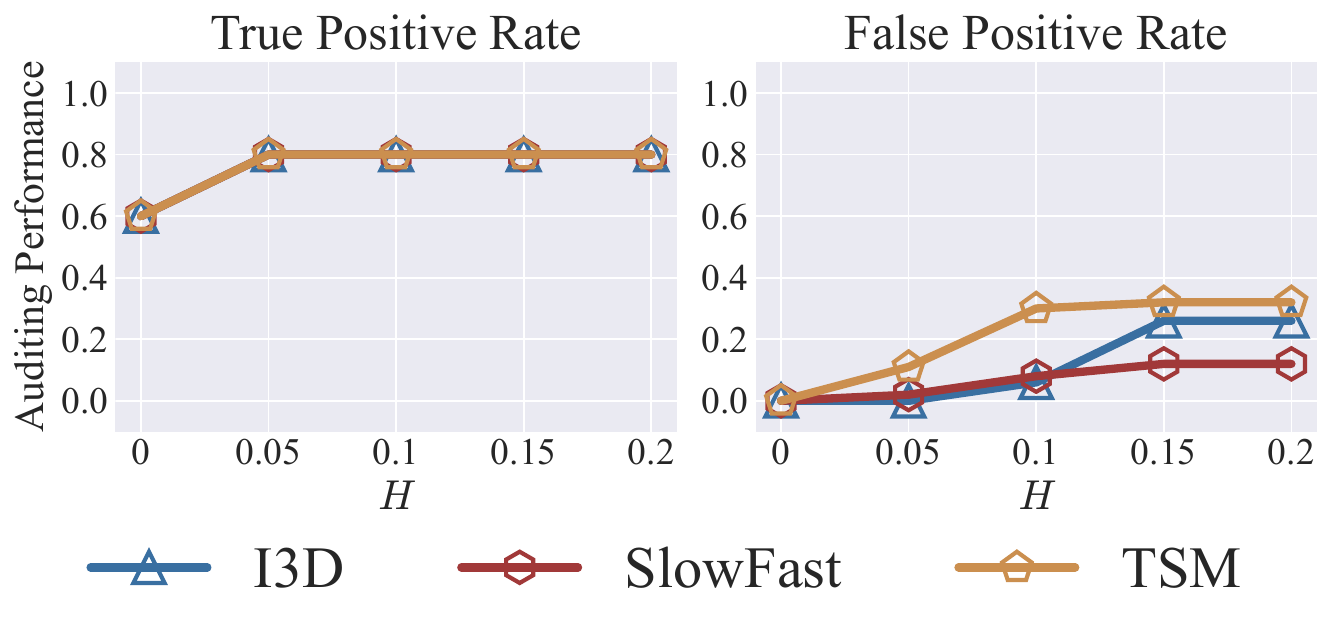}
    \vspace{-0.2cm}
    \caption{
    The impact of different threshold uppers for three various suspect models on the two metrics of the UCF-101 dataset when the noise perturbation budget $\varepsilon=4$.
    }
    \label{fig:param_vary_clip_th_ucf101}
\vspace{-0.4cm}
\end{figure}

\begin{figure}[!htbp]
    \centering
    \includegraphics[width=0.35\textwidth]{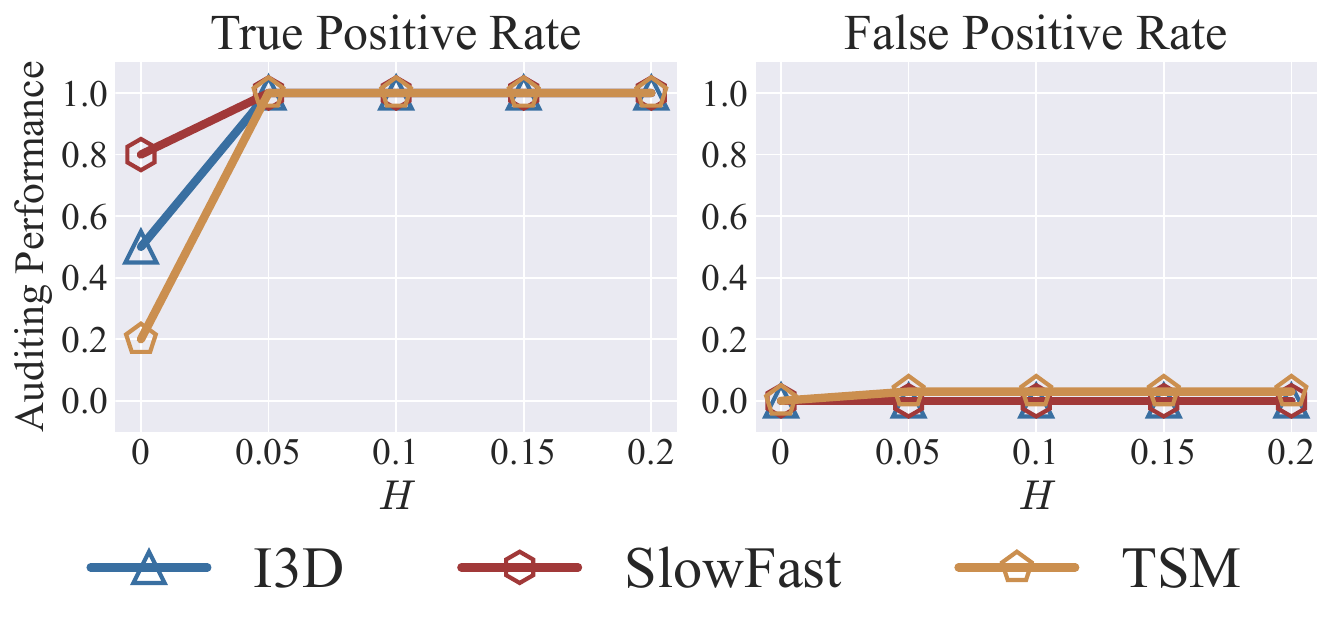}
    \caption{
    The impact of different threshold uppers for three various suspect models on the two metrics of the HMDB-51 dataset when the noise perturbation budget $\varepsilon=4$.
    }
    \label{fig:param_vary_clip_th_hmdb51}
\vspace{-0.3cm}
\end{figure}

\autoref{fig:param_vary_clip_th_hmdb51} presents the TPR and FPR results under different threshold upper bounds on the HMDB-51 dataset. Consistent with the trends observed on the UCF-101 dataset, we found that when the threshold upper bound $H$ is set to a very small value (\eg, 0), the FPR remains very low, but the TPR drops significantly. 
As the threshold increases, the TPR improves notably, but at the cost of a higher FPR. 
To strike a good balance between TPR and FPR across different datasets and suspect models, we consistently set $H=0.05$ in the experiments.

\subsection{Results on Larger Dataset and Backbones}
\label{subsec:result_larger_dataset_backbone}

In this section,
we explore the effectiveness of \method on a larger dataset (\ie, SSv2) and the transformer-based backbone (\ie, TimeSformer).

\autoref{table:ssv2_audit} provides the overall auditing performance of \method.
It can be seen that \method still achieves perfect auditing results on the large datasets, which highlights the versatility of \method.
Furthermore,
we also explore the auditing performance under 
Top-$K$ ($K=5$) and Label-only settings.
As shown in~\autoref{table:ssv2_topk_label},
\method still exhibits great auditing accuracy.
\method achieves 100\% auditing accuracy with the Top-$K$ setting, while a slight decrease in auditing performance occurs with the Label-only setting.
This indicates that \method remains robust to stricter constraints. 

\begin{table}[htbp]
    \centering
    \caption{Auditing performance for the TimeSformer model on the SSv2 dataset.}
    \label{table:ssv2_audit}
    \footnotesize
    \setlength{\tabcolsep}{1.2em}
	\begin{tabular}{c|cccc}
		\toprule
		\textbf{Metric} & TPR & FPR & F1 & $\Delta$acc  \\
		\midrule
        \textbf{Result} & 1.000 & 0.000 & 1.000 & $-0.003$ \\
        \bottomrule
	\end{tabular}
    \vspace{-0.2cm}
\end{table}

\begin{table}[htbp]
    \centering
    \caption{Auditing performance for the TimeSformer model on the SSv2 dataset under Top-$K$ and Label-only settings.}
    \label{table:ssv2_topk_label}
    \footnotesize
    \setlength{\tabcolsep}{1.2em}
	\begin{tabular}{c|cc}
		\toprule
		\textbf{Setting} & TPR & FPR   \\
		\midrule
        \textbf{Top-$K$} & 1.000 & 0.000  \\
        \textbf{Label-only} & 1.000 & 0.001  \\
        \bottomrule
	\end{tabular}
    \vspace{-0.2cm}
\end{table}

\subsection{Robustness under Common Perturbations}
\label{subsec:robustness_common_pert}
In real deployment,
there exist a series of common perturbations that may affect the auditing performance.
In this section,
we explore the robustness of \method under three types of common perturbations.

\subsubsection{Output Quantization}
In this setting,
the attacker tries to reduce information leakage by quantifying precise probability values.
The output probabilities are rounded to one decimal place.
\autoref{fig:output_quantization} illustrates the auditing performance under output quantization on the HMDB-51 and UCF-101 datasets.

It can be seen that \method still achieves strong auditing performance across multiple suspect models and datasets.
The reason is that the quantized results still reflect the behavioral characteristics of model,
which serves as the crucial basis in our dataset auditing. 
This also shows the robustness of our proposed method.

\begin{figure}[htbp]
    \centering
    \includegraphics[width=0.35\textwidth]{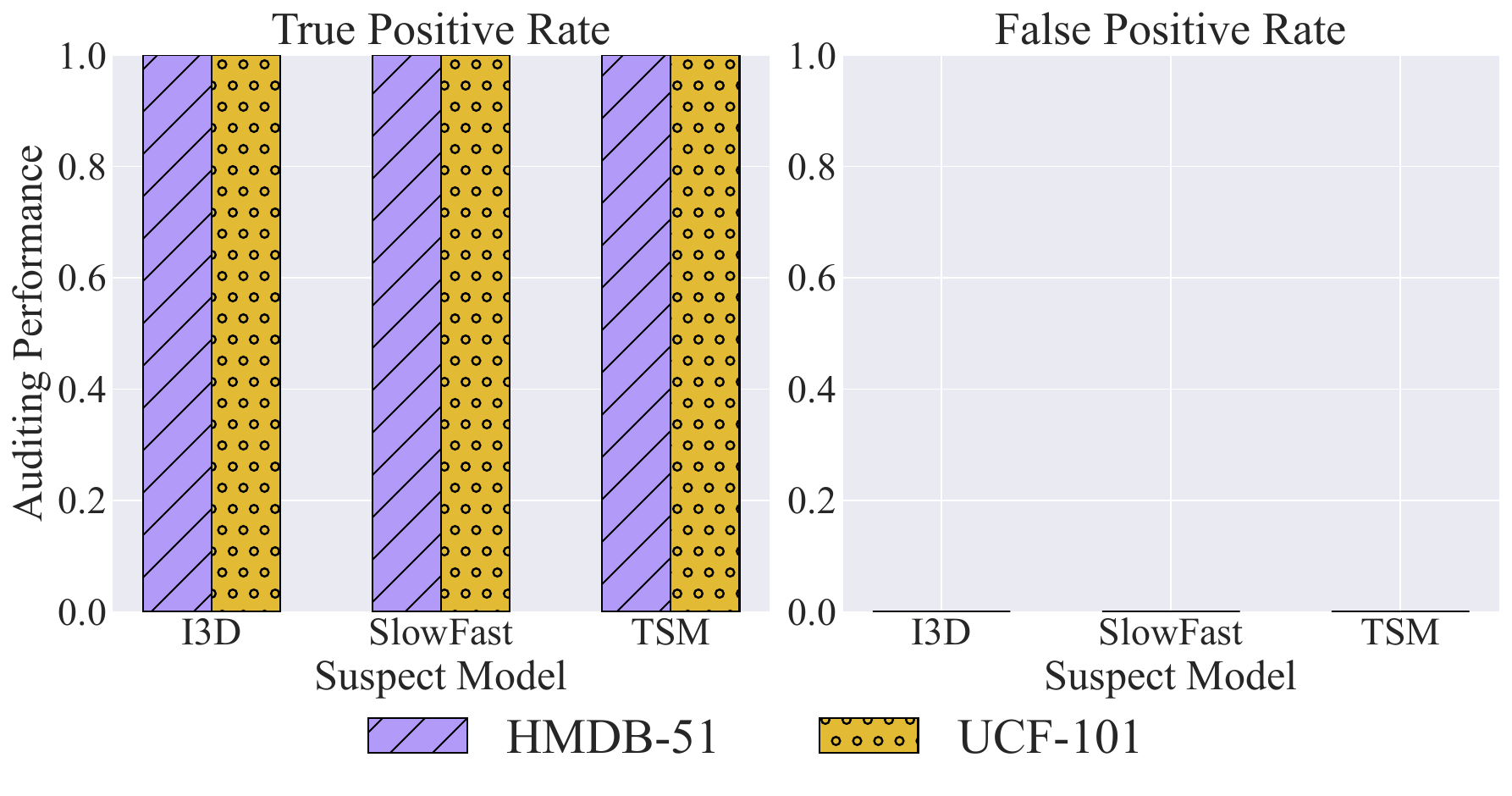}
    \vspace{-0.3cm}
    \caption{
    The auditing performance under output quantization. %
    }
    \label{fig:output_quantization}
\vspace{-0.3cm}
\end{figure}

\subsubsection{Limited Query}
In this scenario, the number of queries performed on the suspect model is limited.
Considering that the HMDB-51 and UCF-101 datasets require a relatively low number of queries (\ie, less than 100 or approximately 250), here we choose to explore the auditing effect on the SSv2 dataset with limited queries.

\autoref{fig:query_limit_ssv2} provides the auditing performance under limited query settings.
We find that \method exhibits promising performance under a small query limit (\eg, 200).
Even for a smaller query limit (\ie, 100),
\method can still identify most cases of dataset misuse and achieve a low FPR.
Therefore, 
\method remains robust under limited query constraints.

\begin{figure}[htbp]
    \centering
    \includegraphics[width=0.2\textwidth]
    {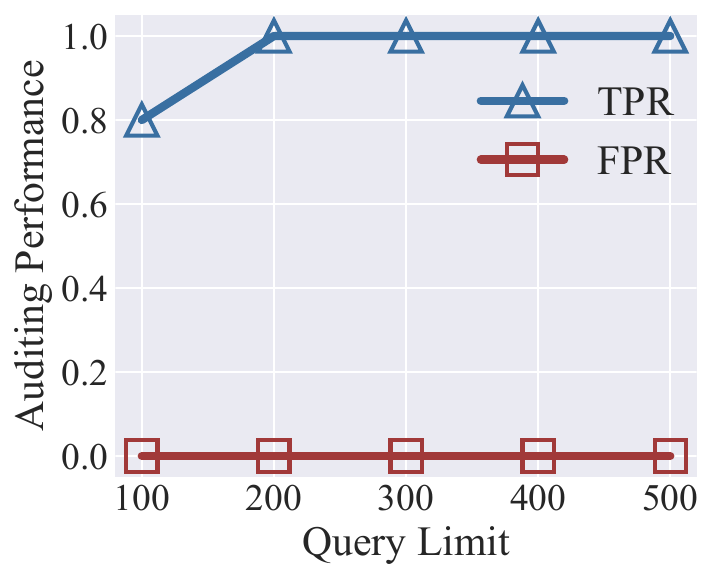}
    \vspace{-0.2cm}
    \caption{
    The auditing performance under limited query settings of the SSv2 dataset.
    }
    \label{fig:query_limit_ssv2}
\vspace{-0.3cm}
\end{figure}

\subsubsection{Lossy Re-encoding}
In our test pipeline,
common video preprocessing steps (\eg, cropping and sampling) are already included.
Here, another common step (\ie, lossy compression) is applied to reduce video size, which may have a potential impact on dataset auditing. 
\autoref{table:ssv2_lossy_encode} illustrates the auditing performance for the TimeSformer model on the SSv2 dataset under this setting.
We find that \method still achieves great auditing accuracy in this re-encoding scenario since 
the compressed video still retains the feature information of the original video.
In addition, it can be seen that this compression has a negative impact on accuracy for normal tasks.

\begin{table}[htbp]
    \centering
    \caption{Auditing performance for the TimeSformer model on the SSv2 dataset under lossy re-encoding.}
    \label{table:ssv2_lossy_encode}
    \vspace{-0.1cm}
    \footnotesize
    \setlength{\tabcolsep}{1.2em}
	\begin{tabular}{c|cccc}
		\toprule
		\textbf{Metric} & TPR & FPR & F1 & $\Delta$acc  \\
		\midrule
        \textbf{Result} & 1.000 & 0.000 & 1.000 & $-0.015$ \\
        \bottomrule
	\end{tabular}
    \vspace{-0.2cm}
\end{table}

\subsection{Robustness Against Adaptive Attackers}
\label{subsec:robustness_adaptive_attack}

In this section,
we consider an adaptive and stronger attacker,
which can adopt a series of procedures to evade auditing.
Here,
we implement four types of adaptive attack settings (denoted as A1-A4) to verify the robustness of \method.
In particular,
A1 denotes the content-aware denoising (\ie,  bilateral filter),
and the neighborhood diameter of the filter is set to $5$.
A2 denotes the mild blur,
and the kernel size is $(5, 5)$.
A3 denotes the frequency-domain suppression (\ie, low-pass), 
and the normalized cutoff frequency is $0.1$.
A4 denotes the adversarial adaptation (\ie, adversarial training),
and 10\% of the published samples are randomly selected and injected with similar Perlin noise.

\autoref{table:ssv2_adaptive_attack} provides the auditing performance under various adaptive attack settings.
\method achieves 100\% audit accuracy in all four adaptive settings. 
The reason is that these removal operations do not completely eliminate the memory of suspect models on the published samples. 
In this case, the model's behavioral differences between published and unpublished samples remain significant, which can be utilized as a basis for successful auditing.
We observe that these adaptive attack methods exert a more pronounced impact on the performance of normal training than using the original modified samples. 
Moreover, compared with the HMDB-51 and UCF-101 datasets, different training methods exhibit smaller performance variations on normal tasks when evaluated on SSv2. 
This can be attributed to the larger scale of SSv2, which promotes better generalization performance.

\begin{table}[htbp]
    \centering
    \caption{Auditing performance for the TimeSformer model on the SSv2 dataset under adaptive attack settings.}
    \label{table:ssv2_adaptive_attack}
    \vspace{-0.1cm}
    \footnotesize
    \setlength{\tabcolsep}{1.2em}
	\begin{tabular}{c|cccc}
		\toprule
		\textbf{Type} & TPR & FPR & F1 & $\Delta$acc  \\
		\midrule
        \textbf{A1} & 1.000 & 0.000 & 1.000 & $-0.009$ \\
        \textbf{A2} & 1.000 & 0.000 & 1.000 & $-0.008$ \\
        \textbf{A3} & 1.000 & 0.000 & 1.000 & $-0.013$ \\
        \textbf{A4} & 1.000 & 0.000 & 1.000 & $-0.011$ \\
        \bottomrule
	\end{tabular}
    \vspace{-0.2cm}
\end{table}

\subsection{Efficiency Analysis}
\label{subsec:efficiency_analysis}

In this section, we explore the computation efficiency of \method on various datasets and models.
The experimental results are obtained on a single A6000 GPU.

\autoref{table:ssv2_runtime} illustrates the running time of different phases for \method on multiple datasets, 
and \autoref{table:ssv2_memory} provides the memory consumption when training on various models.
We observe that the running time on the SSv2 dataset is obviously higher than other two datasets since the scale of SSv2 is large.
The phase of evaluation model training consumes the majority of the time.
For large-scale datasets that may be encountered in real-world deployments, the  time of evaluation model training can be effectively reduced by decreasing the training set size, changing the pretrained model, and reducing the number of training epochs.
The results in~\autoref{fig:param_vary_eval_model_ucf101} indicate that different evaluation models exhibit good transferability.
According to~\autoref{fig:ablation_eval_model_ucf101},
\method can still achieve competitive auditing results even without the evaluation model.

\begin{table}[htbp]
    \centering
    \caption{Running time of different phases for \method.}
    \label{table:ssv2_runtime}
    \vspace{-0.1cm}
    \footnotesize
    \setlength{\tabcolsep}{1.2em}
	\begin{tabular}{c| c | c | c }
		\toprule
		\textbf{Phase} & HMDB-51 & UCF-101 & SSv2   \\
		\midrule
        \textbf{Training} & 1 h 50 min & 3 h 46 min & 25 h 12 min  \\
        \textbf{Modification} & 9 min & 26 min & 4 h 2 min  \\
        \textbf{Selection} & 6 min  & 25 min & 4 h 9 min  \\
        \textbf{Verification} & 0.5 min & 2 min & 21 min  \\
        \bottomrule
	\end{tabular}
    \vspace{-0.2cm}
\end{table}

\begin{table}[htbp]
    \centering
    \caption{Memory consumption (unit: Gigabytes).
    }
    \label{table:ssv2_memory}
    \vspace{-0.1cm}
    \footnotesize
    \setlength{\tabcolsep}{1.2em}
	\begin{tabular}{c| c | c | c | c  }
		\toprule
		\textbf{Model} & I3D & SlowFast & TSM  & TimeSformer  \\
		\midrule
        \textbf{Memory} & 9.66 & 9.31 & 10.25 & 16.69  \\
        
        \bottomrule
	\end{tabular}
    \vspace{-0.2cm}
\end{table}

For the sample modification and selection phases,
the running time shown in~\autoref{table:ssv2_runtime} corresponds to the time taken to calculate all samples in the dataset.
For large datasets, the time required for these two phases can be significantly reduced by decreasing the size of the modified samples (\eg, by 20\% of the original dataset) and parallel computation.
Similarly, 
the running time of verification phase can also be decreased by reducing the number of samples used for verification.
The experimental results in~\autoref{fig:query_limit_ssv2} show that \method remains robust under such limited query conditions.
Therefore, the scalability of 
\method is promising. 
Our method can effectively handle dataset auditing of varying sizes.

\begin{figure*}[htbp]
    \centering
    \includegraphics[width=0.95\textwidth]{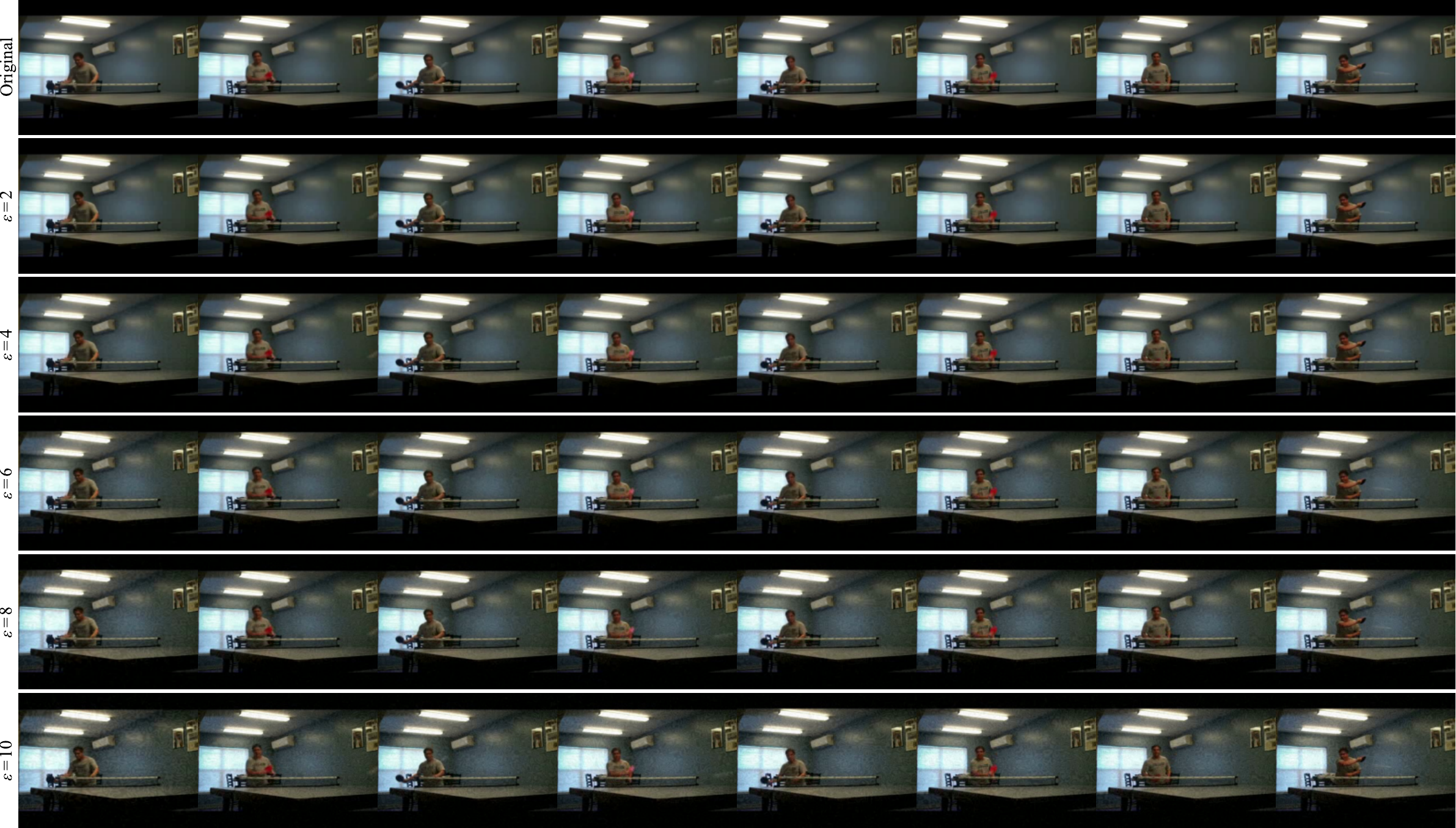}
    \caption{
    An illustration of generated videos under different perturbation budgets for UCF-101.
    }
    \label{fig:vary_epsilon_video_show}
\end{figure*}

\begin{figure*}[htbp]
    \centering
    \includegraphics[width=0.95\textwidth]%
{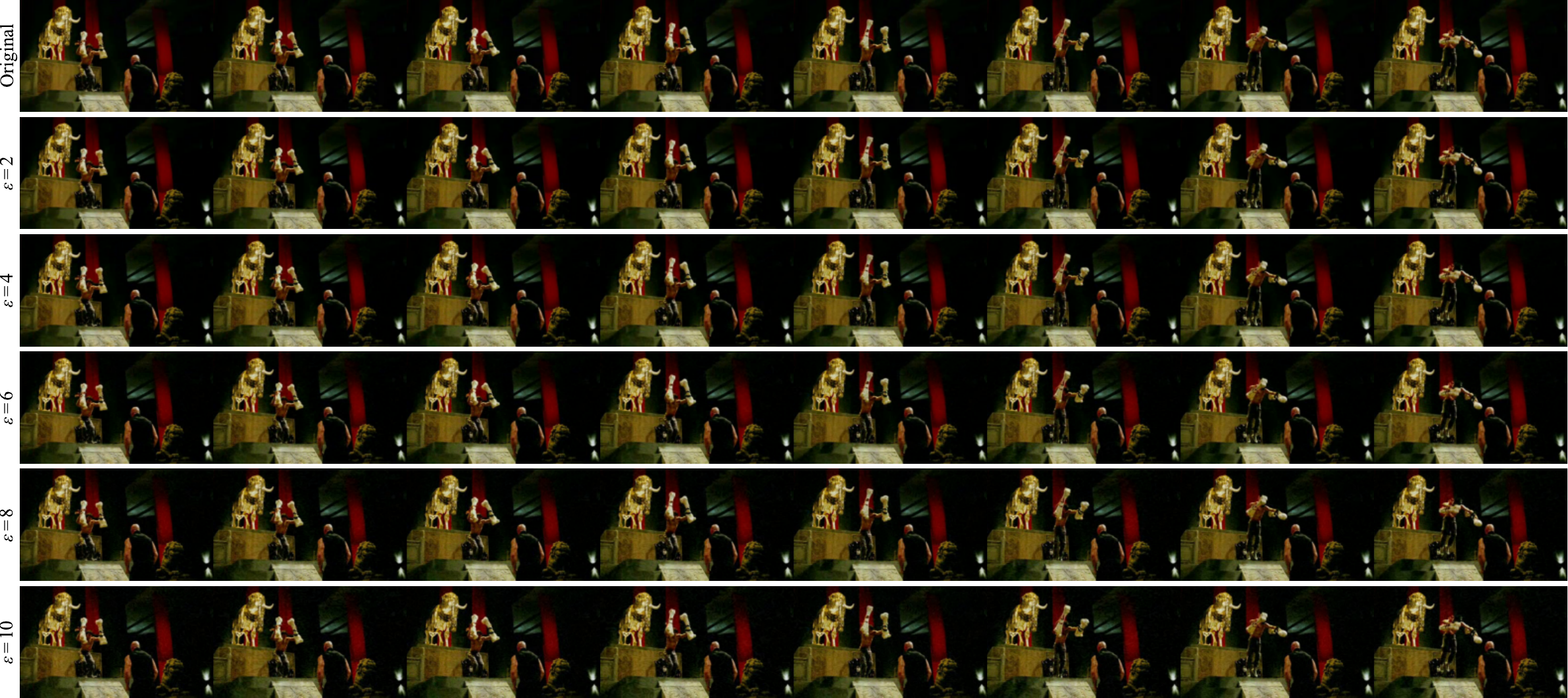}
    \caption{
    An illustration of generated videos under different perturbation budgets for HMDB-51.
    }
    \label{fig:vary_epsilon_video_show2}
\end{figure*}

\begin{figure*}[htbp]
    \centering
    \includegraphics[width=0.95\textwidth]{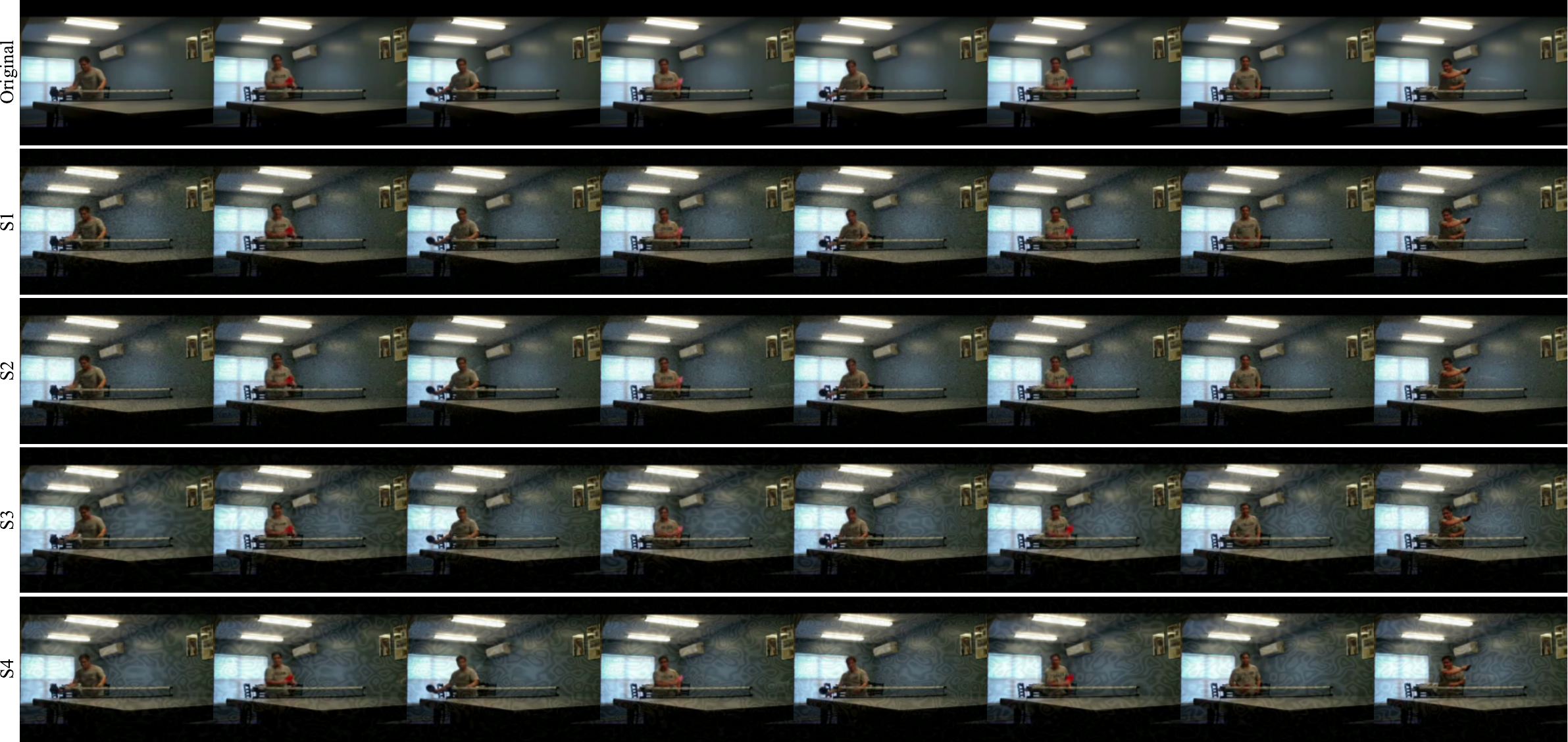}
    \caption{
    An illustration of generated videos under different noise parameter settings for UCF-101.
    }
    \label{fig:vary_perlin_video_show}
\end{figure*}

\begin{figure*}[htbp]
    \centering
    \includegraphics[width=0.95\textwidth]%
{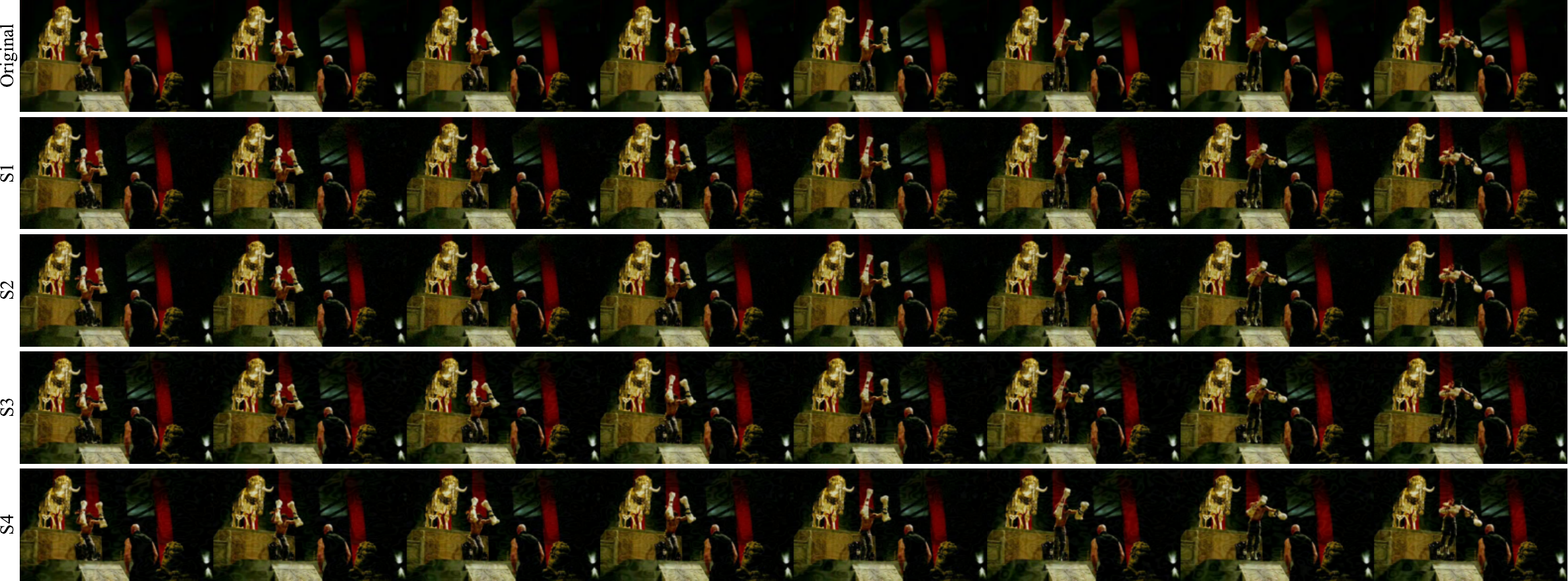}
    \caption{
    An illustration of generated videos under different noise parameter settings for HMDB-51.
    }
    \label{fig:vary_perlin_video_show2}
\end{figure*}

\end{document}